\documentclass[11pt]{article}

\usepackage{url}
\usepackage{xcolor}
\usepackage{color}
\usepackage{graphicx}
\usepackage{amssymb}
\usepackage{amsmath,amsthm,paralist}
\usepackage{amsbsy}
\usepackage{cite}       
\usepackage{algorithm,algorithmic}
\usepackage[margin=1in]{geometry}

\usepackage{epstopdf}

\usepackage{comment}

\newcommand{\editnew}[1]{{#1}} 

\allowdisplaybreaks

\def\real    { \mathbb{R} }

\newtheorem{cor}{Corollary}

\newtheorem{definition}{Definition}

\newtheorem{theorem}{Theorem}
\newtheorem{lem}{Lemma}
\theoremstyle{remark}

\newcommand{\bitem}{\begin{itemize}}
\newcommand{\eitem}{\end{itemize}}

\newcommand{\supp}{\mathrm{supp}}
\newcommand{\beqn}{\begin{equation}}
\newcommand{\eeqn}{\end{equation}}
\newcommand{\balign}{\begin{align}}
\newcommand{\ealign}{\end{align}}

\newcommand{\inner}[1]{\left<#1\right>}
\newcommand\numberthis{\addtocounter{equation}{1}\tag{\theequation}}
\def\x{{\mathbf x}}
\def\y{{\mathbf y}}

\def\a{{\mathbf a}}

\def\u{{\mathbf u}}

\def\z{{\mathbf z}}

\def\z{{\mathbf z}}

\def\A{{\mathbf A}}

\def\F{{\mathbf F}}
\def\H{{\mathbf H}}
\def\S{{\mathbf S}}
\def\bI{{\mathbf I}}
\def\I{{\mathbf I}}

\def\Pr{{\mathbb P}}

\newcommand{\oper}[1]{\mathcal{#1}}

\def \C {\mathbb{C}}
\DeclareMathOperator*{\argmin}{arg min}

\DeclareMathOperator*{\trace}{tr}

\newcommand\blfootnote[1]{
  \begingroup
  \renewcommand\thefootnote{}\footnote{#1}
  \addtocounter{footnote}{-1}
  \endgroup
}

\renewcommand{\ss}{{\vspace*{-1mm}}}

\begin{document}

\title{Constrained adaptive sensing}

\author{Mark~A.~Davenport,
        Andrew~K.~Massimino,
        Deanna~Needell,
        and~Tina~Woolf}
        
\maketitle
       
\begin{abstract} Suppose that we wish to estimate a vector $\x \in \C^n$ from a small number of noisy linear measurements of the form $\y = \A \x + \z$, where $\z$ represents measurement noise.  When the vector $\x$ is sparse, meaning that it has only $s$ nonzeros with $s \ll n$, one can obtain a significantly more accurate estimate of $\x$ by adaptively selecting the rows of $\A$ based on the previous measurements provided that the signal-to-noise ratio (SNR) is sufficiently large. In this paper we consider the case where we wish to realize the potential of adaptivity but where the rows of $\A$ are subject to physical constraints.  In particular, we examine the case where the rows of $\A$ are constrained to belong to a finite set of allowable measurement vectors. We demonstrate both the limitations and advantages of adaptive sensing in this constrained setting. We prove that for certain measurement ensembles, the benefits offered by adaptive designs fall far short of the improvements that are possible in the unconstrained adaptive setting.  On the other hand, we also provide both theoretical and empirical evidence that in some scenarios adaptivity does still result in substantial improvements even in the constrained setting. To illustrate these potential gains, we propose practical algorithms for constrained adaptive sensing by exploiting connections to the theory of optimal experimental design and show that these algorithms exhibit promising performance in some representative applications.
\end{abstract}

\blfootnote{M.~A.~Davenport and A.~K.~Massimino are with the School
of Electrical and Computer Engineering, Georgia Institute of Technology, Atlanta,
GA, 30332 USA (e-mail: mdav@gatech.edu, massimino@gatech.edu).}
\blfootnote{D.~Needell is with the Department of Mathematical Sciences, Claremont McKenna College, Claremont, CA, 91711 USA (e-mail: dneedell@cmc.edu).}
\blfootnote{T.~Woolf is with the Institute of Mathematical Sciences, Claremont Graduate University, Claremont, CA, 91711 USA (e-mail: tina.woolf@cgu.edu).}
\blfootnote{The work of M.~A.~Davenport and A.~K.~Massimino was supported by grants NRL N00173-14-2-C001, AFOSR FA9550-14-1-0342, NSF CCF-1409406 and CCF-1350616, and gifts from LexisNexis Risk Solutions and Lockheed Martin. The work of D.~Needell and T.~Woolf was
partially supported by NSF Career DMS-1348721 and the Alfred P. Sloan Foundation. }

\section{Introduction}

Suppose that we wish to estimate a sparse vector from a small number of noisy
linear measurements. In the setting where the measurements are selected in
advance (independently of the signal) we now have a rich understanding of both
practical algorithms and the theoretical limits on the performance of these
algorithms. A typical result from this literature states that for a suitable
measurement design, one can estimate a sparse vector with an accuracy that
matches the minimax lower bound up to a constant factor~\cite{CandeD_How}.
Such results have had a tremendous impact in a variety of practical settings.
In particular, they provide the mathematical foundation for ``compressive
sensing,'' a paradigm for efficient sampling that has inspired a range of new sensor designs over the last decade.

A distinguishing feature of the standard compressive sensing paradigm is that the measurements are {\em nonadaptive}, meaning that a fixed set of measurements are designed and acquired without allowing for any possibility of adapting as the measurements begin to reveal the structure of the signal. While this can be attractive in the sense that it enables simpler hardware design, in the context of sparse estimation this also leads to some clear drawbacks.  In particular, this would mean that even once the acquired measurements show us that portions of the signal are very likely to be zero, we may still expend significant effort in ``measuring'' these zeros! In such a case, by {\em adaptively} choosing the measurements, dramatic improvements may be possible.

Inspired by this potential, 
recent investigations have shown that we can often acquire a sparse (or
compressible) signal via far fewer measurements or far more accurately if we
choose them adaptively~(e.g.,
see~\cite{DavenA_Compressive,IndykPW_Power,HauptCN_Distilled2,MalloN_Nearb}).
This body of work, which will be discussed in greater detail in Section~\ref{sec::adaptive}, demonstrates that adaptive sensing indeed offers the potential for dramatic improvements over nonadaptive sensing in many settings. However, the existing approaches to adaptive sensing, which rely on being able to acquire {\em arbitrary} linear measurements, cannot be applied in most real-world applications where the measurements must respect certain physical {\em constraints}.  In this paper, our focus is on {\em constrained adaptive sensing}, where our measurements are restricted to be chosen from a particular set of allowable measurements. We will see that new algorithms are required and explore the theoretical limits within this more restrictive setting.  Before describing the constrained adaptive setting in more detail, we first provide a brief review of existing approaches to nonadaptive and adaptive sensing of sparse signals.

\subsection{Nonadaptive sensing} \label{sec::nonadaptive}

In the standard nonadaptive compressive sensing
framework~\cite{CandeRT_Robust,Cande_Compressive,Donoh_Compressed,DavenDEK_Introduction},
we acquire a signal $\x$ via the linear measurements $\y = \A \x + \z$, where
$\A$ is an $m \times n$ matrix representing the sensing system and $\z$
represents measurement noise. The goal is to design $\A$ so that $m$ is
smaller than $n$ by exploiting the fact that $\x$ is {\em sparse} (or nearly sparse).  
Given a basis 
$\mathbf{\Psi}$, we say that a signal $\x \in \mathbb{C}^n$ is $s$-sparse if it can be represented by a linear combination of just $s$ elements from $\mathbf{\Psi}$, i.e., we can write $\x = \mathbf{\Psi} \boldsymbol{\alpha}$, with $\|\boldsymbol{\alpha}\|_0 \le s$, where $\|\boldsymbol{\alpha}\|_0 := |\mbox{supp}(\boldsymbol{\alpha})|$ denotes the number of nonzeros in $\boldsymbol{\alpha}$.  We will typically be interested in the case where $s \ll n$.

There is now a rich literature that describes a wide range of techniques for
designing an appropriate $\A$ and efficient algorithms for recovering $\x$.
In much of this literature, the matrix $\A$ is chosen via randomized
constructions that are known to satisfy certain desirable properties such as
the so-called {\em restricted isometry property} (RIP).\footnote{See Section~\ref{sec::lower bounds} for a more detailed discussion of the RIP
and its implications in the context of adaptive sensing. Note that the RIP is typically stated to require $\|\A \x\|_2 \approx \| \x \|_2$ for all $s$-sparse $\x$, which implies a fixed scaling for the matrix $\A$ where $\|\A\|_F^2 \approx n$. To ease the comparison with results that arise in contexts with alternative scalings, in the result stated in~\eqref{eq:nonadaptive} we make no assumption on the scaling of $\A$ and merely require $\|\A \x\|_2 \approx \beta \| \x \|_2$ for some $\beta>0$.} Under the assumption
that $\A$ satisfies the RIP (or that $\A \mathbf{\Psi}$ satisfies the RIP in the case where $\mathbf{\Psi} \neq \I$), if each entry of $\z$ is independent white Gaussian noise with variance $\sigma^2$ then one can show that techniques based on $\ell_1$-minimization 
produce an approximation $\widehat{\x}$ satisfying
\begin{equation} \label{eq:nonadaptive}
\mathbb{E} \| \widehat{\x} - \x \|_2^2 \le C \frac{n \log n}{\| \A \|_F^2} s \sigma^2,
\end{equation}
where $C > 1$ is a fixed constant (e.g., see~\cite[pp.\ 35]{DavenDEK_Introduction}). Note that this bound holds for {\em any} $\x$, and hence any SNR (even the worst-case). It is possible to obtain improved bounds that eliminate the $\log n$ factor when one assumes that the SNR is sufficiently large to ensure that the support is exactly recovered.

One can show that this result is essentially optimal  in the sense that there is no alternative method to choose $\A$ or perform the reconstruction that can do better than this (up to the precise value of the constant $C$)~\cite{CandeD_How}.  In the event that the signal $\x$ is not exactly $s$-sparse, it is also possible to extend these results by introducing an additional term in the error bound that measures the error incurred by approximating $\x$ as $s$-sparse.  See \cite{DavenDEK_Introduction} and references therein for further details.

\subsection{Adaptive sensing}\label{sec::adaptive}

A defining feature of the approach described above is that it is completely
nonadaptive.  When we consider the effect of noise, this nonadaptive approach
might draw some severe skepticism. To see why, note that in the nonadaptive
scenario, most of the ``sensing energy'' is used to measure the signal at
locations where there is no information, i.e., where the signal vanishes.
Specifically, one consequence of using the randomized constructions for $\A$ typically considered in the literature, or alternatively, any matrix satisfying the RIP,  is that the
available sensing energy (i.e.,  $\|\A\|_F^2$)  is evenly distributed across all possible indices.  This is natural since {\em a priori} we do not know where the nonzeros may lie, however, since most of the coordinates $\x_j$ are zero, it also means that the vast
majority of the sensing energy is seemingly wasted.
In other words, by design, the sensing vectors are approximately orthogonal to the signal, yielding a poor signal-to-noise ratio (SNR).

The idea behind adaptive sensing is that we should focus our sensing energy on locations where the signal is nonzero in order to increase the SNR, or equivalently, not waste sensing energy. In other words, one should try to learn as much as possible about the signal while acquiring it in order to design more effective subsequent measurements. Roughly speaking, one would like to $(i)$ detect those entries which are nonzero or significant, $(ii)$ progressively concentrate the sensing vectors on those entries, and $(iii)$ estimate the signal from such localized linear functionals.
Such a strategy is employed by the {\em compressive binary search} and
\emph{compressive adaptive sense and search} strategies
of~\cite{DavenA_Compressive} and~\cite{MalloN_Nearb}.  These
algorithms operate by examining successively smaller
pieces of the signal to accurately determine
the locations of signal energy.
These techniques can yield dramatic improvements in recovery accuracy.

To quantify the potential benefits of an adaptive scheme, suppose that we observe
\begin{equation} \label{eq:measmodel}
y_i = \langle \a_i, \x \rangle + z_i
\end{equation}
where the $z_i$ are independent and identically distributed (i.i.d.) $\mathcal{N}(0,\sigma^2)$ entries and the $\a_i$
are allowed to depend on the measurement history $((y_1, \a_1),\cdots,(y_{i-1}, \a_{i-1}))$, with the only constraint being that $\sum_i \|\a_i\|_2^2 = \| \A \|_F^2$ is fixed.   Consider a simple procedure that uses half of the sensing energy in a nonadaptive way to identify the support of an $s$-sparse vector $\x$ and then adapts to use the remaining half of the sensing energy to estimate the values of the nonzeros.  If such a scheme identifies the correct support, then it is easy to show that this procedure can yield an estimate satisfying
\begin{equation} \label{eq:adaptPotential}
\mathbb{E} \|\widehat{\x} - \x\|_2^2 = \frac{2s}{\|\A\|_F^2} \, s  \sigma^2 .
\end{equation}
If we contrast this result to that in~\eqref{eq:nonadaptive}, which represents the best possible performance in the nonadaptive setting, we see that this simple adaptive scheme can potentially improve upon the nonadaptive scheme by a factor of roughly
$ (n / s) \log n$, which represents a {\em dramatic} improvement in the typical scenario where $s \ll n$.  Of course, this is predicated on the assumption that the first stage of support identification succeeds, which is not always the case.

A fundamental question is thus: {\em in practice, how much lower can the mean squared error (MSE) be when we are allowed to sense the signal adaptively}? The answer is a subtle one.
In~\cite{AriasCD_Fundamental}, it is shown that  there is a fixed constant $C>0$ such that
\begin{equation} \label{eq:adaptLB}
\inf_{\widehat{\x}} \sup_{\|\x \|_0 \leq s} \mathbb{E} \| \widehat{\x}-\x\|_2^2 \ge C \frac{n}{\|\A\|_F^2} \, s  \sigma^2.
\end{equation}
\editnew{In other words, for even the best possible adaptive scheme there are $s$-sparse vectors for which our recovery error is bounded below by \eqref{eq:adaptLB}.}
This lower bound improves upon the nonadaptive performance~\eqref{eq:nonadaptive} by only a factor of $\log n$, coming far short of the improvement that~\eqref{eq:adaptPotential} indicates might be possible.  Similar results are also obtained in~\cite{Castr_Adaptive}.
\editnew{These results are established by considering vectors that are so difficult to estimate that it is impossible to obtain a reliable estimate of their support, and so adaptive algorithms offer limited room for improvement over nonadaptive ones.}

The result (\ref{eq:adaptLB}) does not say that adaptive sensing {\em never} helps. In fact, in practice it {\em almost always} does help.  For example, when some or most of the nonzero entries in $\x$ are only slightly larger than the worst-case amplitude identified in~\cite{AriasCD_Fundamental}, we {\em can} detect them sufficiently reliably to enable the dramatic improvements predicted in~\eqref{eq:adaptPotential}.
More concretely, provided that $\sigma^2$ is not too large\footnote{For example, the compressive binary search procedure
proposed in \cite{DavenA_Compressive}
succeeds in finding the location of the smallest nonzero entry of amplitude
$\mu$ with
probability $1-\delta$ when $\mu^2/\sigma^2 > 16n\log(\frac{1}{2\delta}+1)/\|A \|_F^2$. The result for the procedure  in~\cite{MalloN_Nearb} is similar.}
relative to the nonzero entries of $\x$,
a well-designed adaptive scheme, where the $\a_i$  are chosen sequentially as in~\cite{DavenA_Compressive,MalloN_Nearb}, can achieve
\begin{equation} \label{eq:adaptive}
\mathbb{E} \| \widehat{\x} - \x \|_2^2 \le C' \frac{s}{\|\A\|_F^2} s \sigma^2
\end{equation}
for a fixed constant $C'$, which represents an enormous improvement when $s \ll n$, and demonstrates that the potential benefits suggested in (\ref{eq:adaptPotential}) can be realized in certain regimes.

We briefly note that these results are somewhat reminiscent of classical results from the field of information based complexity~\cite{bakhvalov_operApprox,woznia_Survey,Novak_Power,Donoh_Compressed} as well as more recent results in active learning~\cite{CastrW_Faster}.  Although this literature considers different observation models (e.g., noise-free observations of non-sparse signals), the general theme is that adaptivity is beneficial only in certain regimes (e.g., see~\cite{IndykPW_Power}).  In another direction, we also note that several authors have previously suggested Bayesian approaches to adaptive sensing that are highly relevant to the problems we study in this paper, but which currently lack much in the way of theoretical justification or understanding~\cite{CastrHNR_Finding,JiXC_Bayesian,NickiS_Compressed}.

\begin{figure*}[t]
\centering
\begin{tabular}{cc}
\includegraphics[height=2.3in,width=3in]{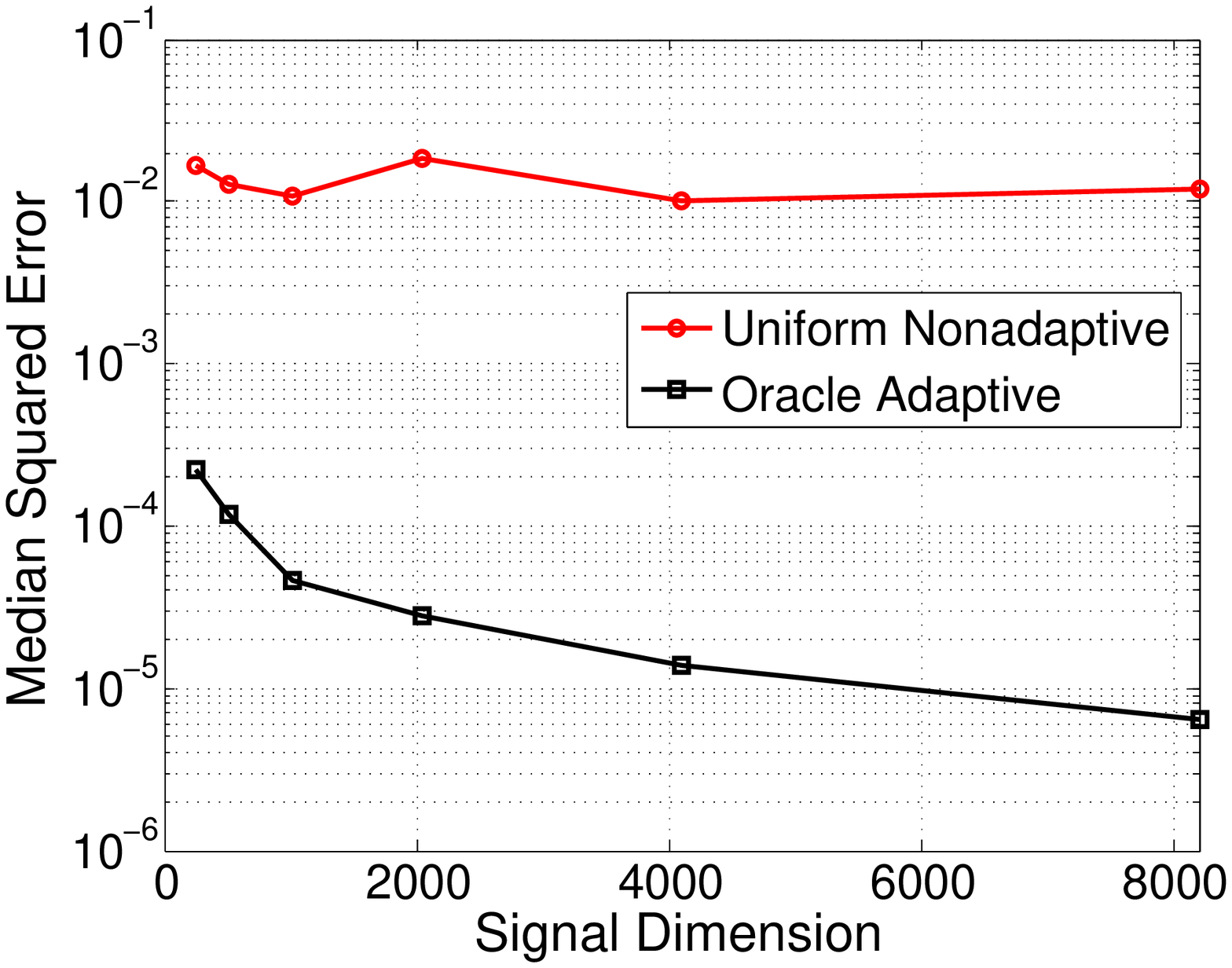} &
\includegraphics[height=2.3in,width=3in]{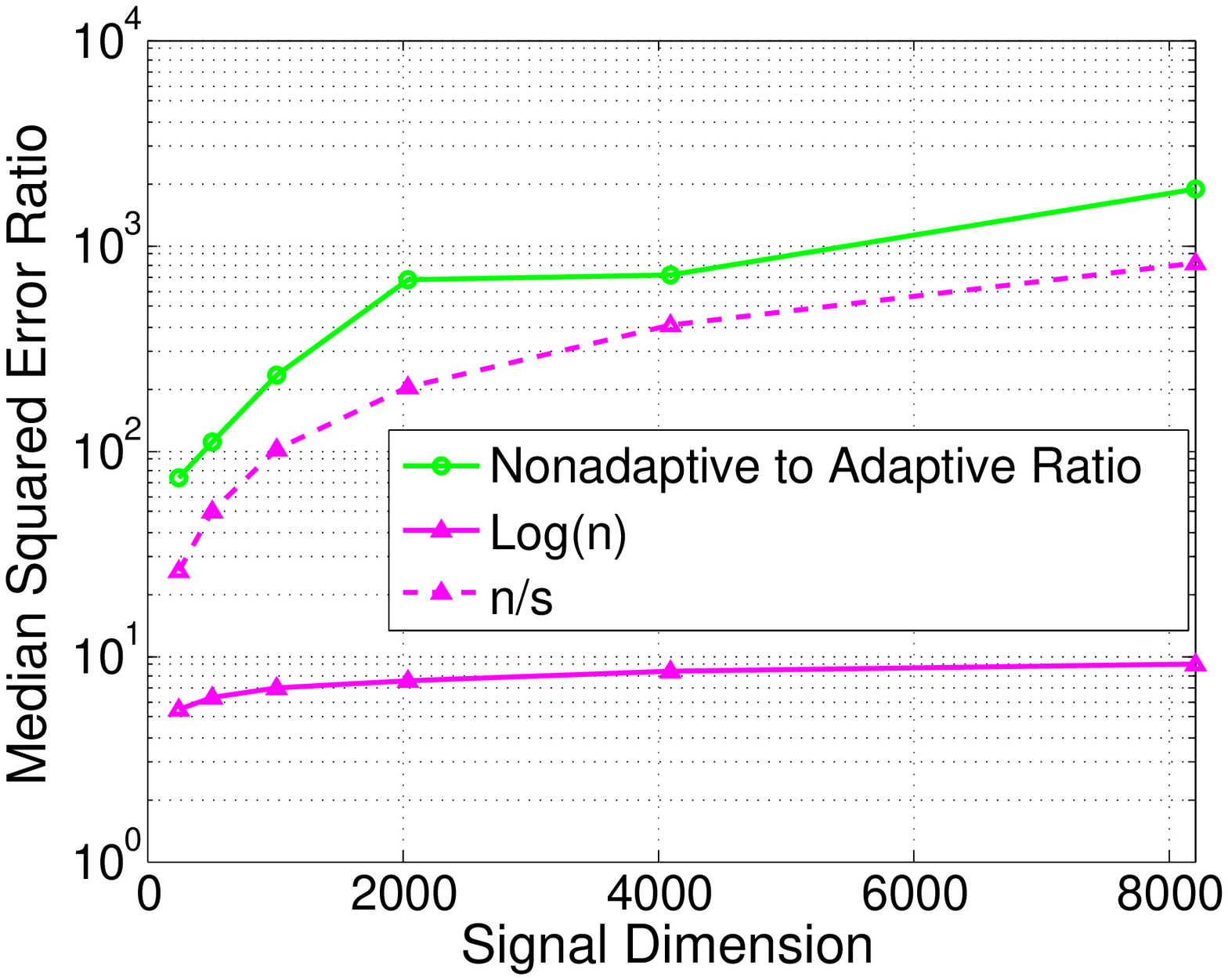}
\end{tabular}
\caption{\small{(Left) The median squared error versus the signal dimension
$n$ for nonadaptive recovery with uniformly random selected measurements (red)
and oracle adaptive recovery (black). (Right) The ratio (green) of the
nonadaptive median squared recovery error to the oracle adaptive median squared recovery error versus the signal dimension $n$, with $\log n$ (solid magenta) and $n/s$ (dashed magenta) included for reference.  }}
\label{fig::UnifVSOracle}
\end{figure*}

\subsection{Constrained sensing} \label{sec::constrained adaptive}

Up to this point, we have discussed results in which we essentially have complete
freedom to design both the adaptive and nonadaptive measurements in an optimal
fashion (that is, up to a constraint on $\|\A\|_F$).  However, there are many applications where such freedom does not exist, and there are significant constraints on the kind of measurements that we can actually acquire. Such constraints arise in various hardware devices inspired by compressive sensing. For example, the single-pixel camera~\cite{DuartDTLSKB_Single} acquires samples of an image by computing inner products with binary patterns.  In this application we could still utilize adaptive measurements, but they must be binary. 
In other applications, we may be restricted to obtaining point samples of the
signal of interest. For example, in standard sampling systems we are
restricted to individually measuring each signal coefficient over time or space. Finally, in tomography and magnetic resonance imaging (MRI), as well as other medical imaging settings, we cannot acquire inner products with arbitrary linear functionals---we are limited to Fourier measurements.

In all of these settings, the measurements are \textit{constrained}; we still
have the flexibility to design measurements adaptively, but we can only select
measurements from a fixed ensemble of predetermined measurements. Thus,
the constrained setting will typically preclude the use of any of the adaptive
sensing algorithms referenced above, and a new approach is required. Specifically, if we
let $\mathcal{M} \subset \C^n$ denote the set of candidate measurement
vectors, then the constrained adaptive sensing problem becomes one of
sequentially selecting the rows $\a_i$ of our sensing matrix from the set
$\mathcal{M}$. In
this work, we assume the multiplicity of a particular measurement from
$\mathcal{M}$ is allowed to be greater than one; that is, repeated measurements are permitted. 
For the methods discussed in this paper, we will restrict our attention to the case where
$\mathcal{M}$ is a finite set.
For a majority of our discussion and examples, we will focus on
the setting where $\mathcal{M} = \{{\bf f}_1,{\bf f}_2,\dots,{\bf f}_n\}$ \editnew{ consists of rows from the Discrete Fourier Transform (DFT) matrix. We stress, however, that we need not require $|\mathcal{M}| = n$
in general.}

With \editnew{the restriction that the measurements be chosen from the DFT ensemble}, Figure~\ref{fig::UnifVSOracle} illustrates the large potential difference between a completely
nonadaptive sensing scheme, where the measurements are selected uniformly at
random, and an ``oracle'' adaptive sensing scheme which uses a priori knowledge of the true locations of the nonzeros in a signal to carefully adapt the choice of measurement vectors to minimize the expected recovery error using the strategy outlined in Section~\ref{sec::opt exp des}.
In both cases, the Compressive Sampling
Matching Pursuit (CoSaMP)~\cite{NeedeT_CoSaMP} algorithm is used for the
signal recovery.
The median\footnote{
Note that the {\it median} and {\it mean} curves exhibit the same overall behavior; however, we display the {median} error across all trials rather than the {mean} error throughout because the median, being a more robust measure, resulted in smoother curves with clearer trends between the methods.} squared error over 200
trials is displayed against the signal dimension $n$. \editnew{Here, the signal is
chosen to have a sparse Haar wavelet decomposition that is supported on a tree.
The choice of a tree-sparse signal is
motivated by the observation that natural images typically have a structured
sparsity pattern in a wavelet domain due to correlations between scales.
In these simulations, the noise level $\sigma^2 = 10^{-4}$ is held constant while the nonzero coefficients scale as $\sqrt{n}$ so that the per-measurement SNR is fixed. The number of measurements taken is set to be $m= 0.6n$ (rounding when necessary).
See Section \ref{subsec:simulations} for further details regarding these simulations.}

It is well-known that the DFT and Haar wavelet transforms are not
incoherent, which implies that $\A {\bf \Psi}$ should not satisfy the RIP; hence, we would
not expect blind nonadaptive sensing to do well in this setting. However, Figure~\ref{fig::UnifVSOracle} does illustrate the large potential for improvement over nonadaptive sensing.  In this case, the adaptive algorithm can potentially improve the recovery error over nonadaptive sensing by roughly a factor of $n/s$, which represents a substantial gain when $s \ll n$.  While we will see below that there are also nonadaptive strategies to address the coherence of Fourier and Haar which somewhat reduce the gap between adaptive and nonadaptive sensing in this case, we believe that this clearly illustrates the potential for adaptive sensing, even in the constrained setting.

\subsection{Organization} The remainder of the paper is organized as follows.
In Section \ref{sec::lower bounds}, we show a simple lower bound on the
adaptive performance of systems limited to \editnew{DFT}
measurements.  We then generalize this result to the larger class of
measurements satisfying the RIP. In both cases, the signal is assumed to be
sparse in the canonical basis. In Section \ref{sec::opt exp des}, we give a
method for measurement selection based on optimal experimental design. In
Section \ref{sec::adapt tomo}, we provide simulations in a more realistic
setting and display numerical results when Fourier measurements are used and
the signal is assumed to be sparse in the Haar wavelet basis,
for both synthetic and realistic signals. We also present some analytical justification using 1-sparse signals in this constrained adaptive setting. Finally, we conclude in Section~\ref{sec::summary} with a brief discussion.

\section{Lower bounds on adaptive performance} \label{sec::lower bounds}

The main result of this section shows that adaptive sensing cannot offer
substantial improvements over the nonadaptive scheme when the measurements are
restricted to certain specific classes of ensembles and the signal is sparse in the canonical basis (i.e., $\mathbf{\Psi} = \mathbf{I}$).  We first consider the
Fourier ensemble, where the sensing vectors are chosen from the
rows of the DFT matrix
$\F\in\C^{n\times n}$, where $\F$ has entries
\begin{equation}\label{dft}
f_{jk} = \frac{1}{\sqrt{n}}\exp(-2\pi\sqrt{-1} jk/n)
\end{equation}
for $j,k = 0, 1, \ldots, n-1$.  In this constrained setting we have the following lower bound.

\begin{theorem}\label{thm1}
\editnew{Under the adaptive measurement model of~\eqref{eq:measmodel}, where the $\a_i$ are chosen (potentially adaptively and allowing repeated measurements) by selecting rows from the DFT matrix~\eqref{dft}, we have that
\begin{equation}\label{dftbound}
\inf_{\widehat{\x}} \sup_{\substack{\|\x \|_0 \leq s \\ \| \x \|_2 \ge R}} \mathbb{E}\|\widehat{\x} - \x\|_2^2 \geq  \frac{n}{m}s\sigma^2
\end{equation}
for any $R \ge 0$.}
\end{theorem}

This shows that even using an optimal choice of sensing vectors, the recovery error is still proportional to $\frac{n}{m}s\sigma^2$, \editnew{{\em even if we exclude the low-SNR setting} (by setting $R$ to be large relative to $\sigma$). This is somewhat reminiscent
of the main results of~\cite{AriasCD_Fundamental} and~\cite{Castr_Adaptive},
which (in an unconstrained setting) establish minimax bounds of the form given in~\eqref{eq:adaptLB}. However, a key difference is that in the unconstrained setting the worst-case error which defines the minimax rate is determined by the performance at a certain range of worst-case SNRs.  Specifically, these bounds are obtained by constructing a ``least favorable prior'' where the nonzeros of $\x$ are near a specific level,\footnote{This threshold is around
$(\min_i x_i^2)/\sigma^2 \approx (n/m) \log s$.} and thus if we were to exclude these challenging $\x$ via the restriction that $\|x\|_2 \ge R$ as in~\eqref{dftbound}, the bound in~\eqref{eq:adaptLB} would be dramatically lower -- in particular, the gains shown in~\eqref{eq:adaptive} could be realized~\cite{DavenA_Compressive,MalloN_Nearb}. Thus, in a sense Theorem~\ref{thm1} is far more pessimistic than these results since it applies {\em no matter how large the SNR} -- although given the incoherence of the DFT and the canonical bases, perhaps this is not that surprising. Finally, we note that for certain values of $R$ it may be possible to obtain a slightly stronger version of Theorem~\ref{thm1} (by a $\log n / \log\log n$ factor) using the techniques in~\cite[Thm 6.1]{hassanieh2012nearly}. We do not pursue these refinements here.}

\begin{proof}[Proof of Theorem~\ref{thm1}]
\editnew{For any adaptive procedure $\widehat{\x}$, we let $\F'$ be the $m\times n$ sensing matrix consisting of the $m$ adaptively chosen vectors from the rows of $\F$, and let $\F'_\Lambda$ denote the $m\times s$ submatrix of $\F'$ whose column indices correspond to the indices of the support $\Lambda$ of $\x$. Using the rows of $\F'$ to acquire the measurements as in~\eqref{eq:measmodel}, we obtain $\y =\F'\x+\z =  \F'_\Lambda\x_\Lambda + \z$. It is not difficult to show (e.g., see the Appendix of~\cite{CandeD_How}) that
\[
\inf_{\widehat{\x}} \sup_{\substack{\|\x \|_0 \leq s \\ \| \x \|_2 \ge R}} \mathbb{E}\|\widehat{\x} - \x\|_2^2 \ge \inf_{\widehat{\x}} \sup_{\substack{\x' \in \real^s \\ \| \x' \|_2 \ge R}} \mathbb{E}\|\widehat{\x}(\F'_\Lambda\x' + \z) - \x'\|_2^2,
\]
where $\widehat{\x}(\cdot)$ takes values in $\real^s$.}

\editnew{To establish the bound in~\eqref{dftbound} we consider a sequence of least favorable prior distributions on $\x'$. The minimax risk is always larger than the Bayes risk under any prior, so this will establish a lower bound on the minimax risk. Towards this end, consider the prior on $\x'$ where $\x' \sim \mathcal{N}(0, \rho^2 \bI)$, but where the distribution is truncated to be zero for $\| \x' \|_2 \le R$ and re-scaled appropriately. Note that in the absence of this truncation, the Bayes risk would be given by
\begin{equation} \label{eq:Bayes}
\sigma^2 \sum_{i=1}^s \left( \frac{\sigma_i(\F'_\Lambda)}{\sigma_i^2(\F'_\Lambda) + \frac{\sigma^2}{\rho^2}} \right)^2,
\end{equation}
where $\sigma_i(\F'_\Lambda)$ denotes the $i^{\text{th}}$ singular value of $\F'_\Lambda$. This follows from the fact that the Bayes estimator is given by \[
\mathbb{E}[\x|\y] = ({\F'_\Lambda}^T \F'_\Lambda + \textstyle{\frac{\sigma^2}{\rho^2}} \bI)^{-1}{\F'_\Lambda}^T \y.
\]
The result in~\eqref{eq:Bayes} follows from the fact that for this estimator the expected squared error is given by 
\[
\sigma^2 \|({\F'_\Lambda}^T \F'_\Lambda + \textstyle{\frac{\sigma^2}{\rho^2}} \bI)^{-1}{\F'_\Lambda}^T\|_F^2
\]
which reduces to~\eqref{eq:Bayes} via the application of standard properties of the singular value decomposition.  We now note that 
for any $R \ge 0$, as $\rho^2 \rightarrow \infty$, the Bayes risk for the truncated prior will converge to that of~\eqref{eq:Bayes}, namely,
\[
\sigma^2 \|({\F'_\Lambda}^T \F'_\Lambda)^{-1}{\F'_\Lambda}^T\|_F^2 = \sigma^2 \sum_{i=1}^s \frac{1}{\sigma_i^2(\F'_\Lambda)}
\]
}

\editnew{Putting this all together, we have that}
\begin{align*}
\inf_{\widehat{\x}} \sup_{\substack{\|\x \|_0 \leq s \\ \| \x \|_2 \ge R}} \mathbb{E}\|\widehat{\x} - \x\|_2^2 & \ge \sigma^2 \sum_{i=1}^s \frac{1}{\sigma_i^2(\F'_\Lambda)} \\
& \geq \sigma^2\frac{s^2}{\sum_{i=1}^s \sigma_i^2(\F'_\Lambda)} \\
& = \sigma^2\frac{s^2}{\|\F'_\Lambda\|_F^2}
\end{align*}
where the second inequality follows from Jensen's inequality.  Since $\|\F'_\Lambda\|_F^2 = \frac{sm}{n}$, this completes the proof.
\end{proof}

Our next result generalizes this type of lower bound to any ensemble whose
submatrices satisfy the RIP with \editnew{overwhelming} probability.  This statement is
significant because it suggests that in some constrained
situations, specifically many commonly studied in compressive sensing,
there is little benefit from adaptivity.
\editnew{Formally, we define an RIP ensemble as follows.
\begin{definition}
Let $m$ be fixed. We say that an $n \times n$ matrix $\A$ with unit-norm rows is an {\em RIP ensemble} if for any $m'\geq m$ a
random $m'\times n$ submatrix $\widetilde{\A}$, whose rows are uniformly chosen without replacement,
satisfies
\begin{equation}\label{rip2}
0.5\frac{m'}{n}\|\u\|_2^2 \leq \|\widetilde{\A}\u\|_2^2 \leq 1.5\frac{m'}{n}\|\u\|_2^2,
\end{equation}
for all $s$-sparse $\u$ with probability $1-\exp(-cn)$ (where $c$ is such that $\exp(-cn) < 1/2n$).
\end{definition}}

Theorem~\ref{thm2} makes rigorous the claim that selecting rows intelligently
from such a matrix yields no substantial improvement over a nonadaptive
scheme.
\begin{theorem}\label{thm2}
\editnew{Under the adaptive measurement model of~\eqref{eq:measmodel}, where the $\a_i$ are chosen (potentially adaptively and allowing repeated measurements) by selecting rows from an RIP ensemble as defined above, we have that
\begin{equation}\label{RIPbound}
\inf_{\widehat{\x}} \sup_{\substack{\|\x \|_0 \leq s \\ \| \x \|_2 \ge R}} \mathbb{E}\|\widehat{\x} - \x\|_2^2 \geq  \frac{sn}{3m^2}s\sigma^2
\end{equation}
for any $R \ge 0$.}
\end{theorem}

We note that one usually anticipates $m$ to be on the order of $s\log n$, in which case this bound becomes
$$
\mathbb{E}\|\widehat{\x} - \x\|_2^2 \geq \frac{sn}{3m s\log n}s\sigma^2
    = \frac{n}{3m\log n}s\sigma^2, 
$$
which is roughly a factor of $\log^2 n$ lower than the upper bound in (\ref{eq:nonadaptive}).
This result shows that the recovery error with any adaptive measurements selected from some standard RIP ensemble again falls short of the possible gains shown in~\eqref{eq:adaptPotential}.

We also note here that the bound in Theorem \ref{thm2} is worse by a factor of
$m/s$ than Theorem \ref{thm1}.  However, we believe this is necessary due to
the fact that the only assumption we place on $\A$ is that \eqref{rip2} holds
with \editnew{overwhelming} probability; this is a much weaker requirement than insisting on DFT measurements as in Theorem \ref{thm1}.  As a motivating example, fix some
subset $\Lambda\subset\{1,\ldots, n\}$ of size $s$.  Construct a matrix $\A$
by setting it to the DFT basis $\F$, with its first row modified in the
following way: on $\Lambda$, multiply each entry by a factor of $C$ where $C^2
= {m/8s}$ and off of $\Lambda$ multiply each entry by a factor $c =
\sqrt{(n-sC^2)/(n-s)}$.  This yields a matrix $\A$ whose rows still have unit
norm.  In addition, one can show that for this new matrix $\A$, the property
\eqref{rip2} still holds with the same probability for $\delta = 5/8$ for any
$(m+1)\times n$ submatrix $\widetilde{\A}$.  Construct an $(m+1)\times n$
matrix $\A'$ with the first row of $\A$ repeated $m-s+2$ times (since at least
$s$ rows need to be unique).  Then one computes that $\|\A'_{\Lambda}\|_F^2 =
\frac{s}{n}(s-1 + \frac{m}{8s}(m-s+2))\gtrsim m^2/n$.  On the other hand, any
matrix of the same size adaptively constructed from the DFT basis $\F$ has
a squared Frobenius norm equal to $s(m+1)/n$.  Thus we may indeed lose an $m/s$ factor because of this weakened assumption.

\begin{proof}[Proof of Theorem~\ref{thm2}]
Let $\A'$ be the $m\times n$ matrix of the adaptively selected rows as in the
theorem. Fix a support set $\Lambda$ of size at most $s$.  \editnew{Let $\A'_{\Lambda}$ be
the restriction of $\A'$ to the support set $\Lambda$.  We will prove the result by showing
a bound on the norm of the rows of $\A'_{\Lambda}$ which we obtain via an argument of contradiction.
To that end,} let $\a^{\star}$ be
the row of $\A$ corresponding to the row of $\A'_{\Lambda}$ with the  greatest
Euclidean norm.  Now consider drawing a random $(m+1) \times n$ submatrix
$\widetilde{\A}$ of $\A$ that contains $\a^{\star}$ as a row.
 Then one can compute that any such
submatrix $\widetilde{\A}$ satisfies $\eqref{rip2}$ (with $m'=m+1$) with
probability at least $1 - \exp(-cn)n/(m+1) > 1- \exp(-cn)n$.  \editnew{Indeed, one sees formally
that
\begin{align*}
    \Pr&(\widetilde{\A} \text{ does not satisfy \eqref{rip2}}
        \mid \a^{\star} \text{ is a row of } \widetilde{\A})\\
    &=\frac{\Pr(\widetilde{\A} \text{ does not satisfy \eqref{rip2} and }
        \a^{\star} \text{ is a row of } \widetilde{\A})}
        {\Pr(\a^{\star} \text{ is a row of } \widetilde{\A})}\\
    &\leq \frac{\Pr(\widetilde{\A} \text{ does not satisfy \eqref{rip2}})}
        {\Pr(\a^{\star} \text{ is a row of } \widetilde{\A})}\\
    &\leq \frac{\exp(-cn)}{(m+1)/n}.
\end{align*}}
Now let
$\widetilde{\A}^c$ be the remainder of the matrix, i.e., all rows of
$\widetilde{\A}$ except row $\a^{\star}$.  Similarly, one
computes that any such matrix $\widetilde{\A}^c$ satisfies \eqref{rip2} (with
$m'=m$) with probability at least $ 1- \exp(-cn)n/(n-m) > 1- \exp(-cn)n$.
Thus \textit{both} of these matrices satisfy \eqref{rip2} with probability at
least $1 - 2\exp(-cn)n > 0$.  For the sake of a contradiction, suppose that
$\|\a^{\star}_{\Lambda}\|_2^2 > 3m/n$.  Observe that the signal $\x\in\real^n$ where
$\x_{\Lambda} = \a^{\star}_{\Lambda}$ and padded with zeros off of the support $\Lambda$ is an $s$-sparse signal.   Then since both matrices satisfy \eqref{rip2}, we must have that
\begin{align*}
\|\widetilde{\A}\x\|_2^2 &= \|\widetilde{\A}^c\x\|_2^2 + |\langle \a^{\star}_{\Lambda}, \x\rangle|_2^2\\
&\geq 0.5\frac{m}{n}\|\x\|_2^2 + \|\x\|_2^4 \\
&> \biggl(0.5\frac{m}{n} + \frac{3m}{n}\biggr) \|\x\|_2^2
\geq \frac{3.5m}{n}\|\x\|_2^2.
\end{align*}
On the other hand, we must also have that
\begin{align*}
\|\widetilde{\A}\x\|_2^2 &\leq 1.5\frac{m+1}{n}\|\x\|_2^2.
\end{align*}
Combining these means that $\frac{3.5m}{n} \leq 1.5\frac{m+1}{n} \leq  \frac{3m}{n}$, which is a contradiction.  Thus, it must be that $\|\a^{\star}_{\Lambda}\|_2^2 \leq 3m/n$.
Since $\a^{\star}_{\Lambda}$ is the largest row of $\A'_{\Lambda}$, we then have that
$$
\|\A'_\Lambda\|_F^2 \leq m\|\a^{\star}_{\Lambda}\|_2^2
\leq \frac{3m^2}{n}.
$$
\editnew{Following the same argument as in the proof of Theorem~\ref{thm1}, we} thus have that
$$
\editnew{\inf_{\widehat{\x}} \sup_{\substack{\|\x \|_0 \leq s \\ \| \x \|_2 \ge R}} \mathbb{E}\|\widehat{\x} - \x\|_2^2} \geq \frac{s^2}{\|\A'_{\Lambda}\|_F^2}\sigma^2 \geq
\frac{sn}{3m^2}s\sigma^2,
$$
which completes the proof.

\end{proof}

\section{Adaptivity through optimal experimental design}\label{sec::opt exp des}

Although there are some settings where constrained adaptive sensing does not offer substantial improvement over the nonadaptive scheme, one can of course ask if there are other settings where notable gains are still possible. In order to address this question, we consider the simplified constrained adaptive sensing problem where we assume the support $\Lambda$ of the signal $\x$ (with respect to the sparsity basis $\mathbf{\Psi}$) is known, or some estimate of the support is provided. How would we choose the measurements to best make use of this information, while still respecting that the measurements are constrained to be from the measurement ensemble $\mathcal{M}$? 

Let $\{\a_i\}_{i=1}^m$ denote a sequence of length $m$ with elements
$\a_i \in \mathcal{M}$ corresponding to the measurements of
$\mathcal{M}$ that are chosen.\footnote{We use a {\it sequence} of elements from $\{1,\dots,|\mathcal{M}|\}$ rather than a {\it subset} to emphasize that the $m$ measurements from $\mathcal{M}$ need not be distinct.  \editnew{Note that in the general adaptive setting the \textit{order} of the measurements is also important; however, in the context of this section there is only one batch of adaptive measurements and thus the order within this batch has no impact.}} Then, denote by ${\A}'$ the $m \times n$
matrix (recall $\mathcal{M} \subset \C^n$) whose $i^{\text{th}}$ row is $\a_i$.  If $\Lambda = \supp(x)$, then it can be shown by following the arguments in the proof of Theorem \ref{thm1} that the optimal MSE satisfies
\begin{equation}\label{eqn::MSEandNorm}
\begin{split}
\mathbb{E} \|\widehat{\x} - \x \|_2^2 &= \| (\A'\mathbf{\Psi}_\Lambda)^\dagger\|_F^2 \sigma^2\\
    &= \trace\left( ((\A'\mathbf{\Psi}_\Lambda)^* \A'\mathbf{\Psi}_\Lambda)^{-1}
    \right)\sigma^2,
\end{split}
\end{equation}
where \editnew{$(\A'\mathbf{\Psi}_\Lambda)^\dagger$ denotes the Moore-Penrose pseudoinverse of $\A'\mathbf{\Psi}_\Lambda$}, $\sigma^2$ is the variance of the noise term as in
(\ref{eq:measmodel}), $\mathbf{\Psi}_\Lambda$ is the submatrix of $\mathbf{\Psi}$ restricted to the columns
indexed by $\Lambda$, and $\A'\mathbf{\Psi}_\Lambda$ is assumed to have full (column) rank.  Our goal is to find a length-$m$ measurement sequence $\{\a_i\}_{i=1}^m$ that minimizes (\ref{eqn::MSEandNorm}), which is equivalent to solving 
\begin{equation}\label{eqn::opt}
\{\widehat{\a}_i\}_{i=1}^m = \argmin_{\left\{\{\a_i \}_{i=1}^m \mid\,
\a_i \in \mathcal{M} \right\}\hspace{-12pt}} \;\; \trace\left( ((\A'\mathbf{\Psi}_\Lambda)^* \A'\mathbf{\Psi}_\Lambda)^{-1} \right),
\end{equation}
where $\A' = \A'(\{\a_i \}_{i=1}^m)$ is constructed 
as described above.  \editnew{Note that an essentially equivalent way to state~\eqref{eqn::opt} (up to a permutation of the measurements) is via the discrete optimization problem}
\begin{equation}\label{eqn::relaxation_review}
\begin{split}
\widehat{\S} = &\argmin_{\substack{\text{\footnotesize{diagonal matrices} }\S\succeq 0 \\ s_{ii}\in\mathbb{Z}^+}} \;\; \trace\left( ((\A\mathbf{\Psi}_\Lambda)^* \S \A\mathbf{\Psi}_\Lambda)^{-1} \right)  \\
& \quad \quad \quad \text{subject to~}  \trace(\S)\leq m, \end{split}
\end{equation}
where $\A$ is the $|\mathcal{M}| \times n$ matrix containing all possible
measurement vectors from $\mathcal{M}$ and $s_{ii}\in\mathbb{Z}^+$ forces each diagonal entry of $\S$ to be a non-negative integer (reflecting the multiplicity of each $\a_i$).
\editnew{Both \eqref{eqn::opt} and \eqref{eqn::relaxation_review} reflect the optimization problem that we would ideally like to solve.  Unfortunately they are computationally demanding discrete optimization problems; hence, we instead consider the relaxation of~\eqref{eqn::relaxation_review}}
\begin{equation}\label{eqn::relaxation}
\begin{split}
\widehat{\S} = &\argmin_{\mbox{\footnotesize{diagonal matrices} }\S\succeq 0} \;\; \trace\left( ((\A\mathbf{\Psi}_\Lambda)^* \S \A\mathbf{\Psi}_\Lambda)^{-1} \right)  \\
& \quad \quad \quad \text{subject to~}  \trace(\S)\leq \editnew{m}, \end{split}
\end{equation}
where the constraint
$\trace(\S)\leq \editnew{m}$ ensures that the resulting ``weighted''
sensing matrix $\sqrt{\S}\A$ 
satisfies the ``sensing energy'' constraint $\|{\sqrt{\S}\A}\|_F^2 \leq \editnew{m}$ 
when the rows of $\A$ 
are normalized. 
Note that this is equivalent to the continuous design for the A-optimality criterion studied
in the optimal experimental design literature~\cite{Pukel_OptDesign}.

Fortunately, \eqref{eqn::relaxation} is a convex problem~\cite{BoydV_Convex} and
can be efficiently solved by a number of methods.
Whereas the problem in (\ref{eqn::opt}) would tell us which measurements and how many of each to use from $\mathcal{M}$, (\ref{eqn::relaxation}) instead tells us, through the diagonal matrix $\widehat{\S}$ of weights, ``how much" of each measurement to use. \editnew{We simply weight each $\a_i$ by $\sqrt{\widehat{s}_{ii}}$, where $\widehat{s}_{ii}$ denotes the $i^{\text{th}}$ element on the diagonal of $\widehat{\S}$.}

If, on the other hand, we use the measurement model where we must choose $m$ unweighted
measurements from $\mathcal{M}$, the practical use of $\widehat{\S}$ from (\ref{eqn::relaxation}) is less obvious. We experimented with several different (though likely sub-optimal) approaches to using the weights in $\widehat{\S}$, and the following method empirically seemed to produce the best results.
In this work, we use a simple sampling scheme to obtain a
discrete design.
Specifically, we draw exactly $m$ measurements,
with replacement, according to the probability mass function
\begin{equation} \label{eq:Spmf}
p_i = \frac{\widehat{s}_{ii}}{m}.
\end{equation}
We guarantee that the resulting matrix $\A'$ is at least rank $s$ by rejecting any construction for which this constraint is not satisfied.
These $m$ measurements then form the rows of the sensing matrix $\A'$.

\section{Case study: Fourier measurements of Wavelet sparse signals}\label{sec::adapt tomo}

The results of Section \ref{sec::lower bounds} demonstrate that adaptive sensing cannot offer substantial improvements over nonadaptive sensing for certain classes of measurement ensembles when the signal is sparse in the canonical basis. 
We next explore the case when $\mathbf{\Psi}$ is instead a wavelet basis and we acquire DFT measurements (this is indeed the setting of Figure \ref{fig::UnifVSOracle}, which suggests dramatic potential improvements from constrained adaptive sensing). This setting serves as a somewhat idealized model for a number of applications in tomography and other medical imaging since physical limitations would entail that we can only acquire DFT measurements, and \editnew{realistic} images are generally sparse with respect to wavelet bases~\cite{Daube_Ten}. In this setting we might receive one DFT measurement at a time, and from those, we can (potentially in real time) request the next DFT coefficient to be measured.

For our first two sets of experiments, we will assume the sparsity basis ${\bf \Psi}$ is the Haar wavelet basis. We will denote the $n \times n$ discrete Haar wavelet transform by $\H$, with entries $h_{jk}$ 
for $j,k = 0,1,\ldots,n-1$ and $n$ is assumed to be some power of 2.  When $j = 0$, we have
\begin{equation} \label{eqn::HaarDef1}
h_{0k} = \frac{1}{\sqrt{n}}.
\end{equation}
For indices $j>0$, we write $j = 2^p + q - 1$, where $p =\lfloor \log_2 j \rfloor$ and $q$ are nonnegative integers, and define
\begin{equation} \label{eqn::HaarDef2}
h_{jk}= \frac{1}{\sqrt{n}}
\begin{cases}
2^{p/2} & \frac{(q-1)n}{2^p} \leq k < \frac{(q-0.5)n}{2^p} \\
-2^{p/2}  & \frac{(q-0.5)n}{2^p} \leq k < \frac{qn}{2^p} \\
0  &  \text{otherwise}.
\end{cases}
\end{equation}
Since, however, the Haar wavelet basis $\H$ is a sparsifying transformation, for a signal (or image) $\x$ we have that $\H \x = \boldsymbol{\alpha}$, with $\|\boldsymbol{\alpha}\|_0 \leq s$. This means $\x = \H^*\boldsymbol{\alpha}$, where $\H^*$ denotes the adjoint of $\H$, for which $\H^*=\H^{-1}$ since $\H$ is unitary. 

With this notation in hand and recalling that $\F$ is the $n\times n$ DFT, (\ref{eqn::MSEandNorm}) becomes
\begin{equation}\label{eqn::MSEandNormFH}
\mathbb{E} \|\widehat{\x} - \x \|_2^2 = \| (\F'\H_\Lambda^*)^\dagger\|_F^2 \sigma^2,
\end{equation}
where $\F'$ is the $m\times n$ sensing matrix consisting of the $m$ adaptively chosen vectors from $\F$ and $\H^*_\Lambda$ is the $n\times s$ submatrix of $\H^*$ restricted to the columns indexed by $\Lambda = \supp(\boldsymbol{\alpha}) = \supp(\H\x)$. Thus, we see that the optimal MSE depends on the correlations of the DFT and Haar basis elements. In a similar manner, in our last experiment, where the signal is an MRI image, we will assume the sparsity basis $\mathbf{\Psi}$ is the Daubechies wavelet with 3 vanishing moments (D6).

We now present a suite of numerical simulations in these settings that employ the relaxation (\ref{eqn::relaxation}) followed by the sampling scheme described in Section~\ref{sec::opt exp des} to select a sequence of $m$ DFT measurement vectors.
We then follow with a short analysis for the simple case of 1-sparse signals.

\begin{figure*}[t]
\centering
\begin{tabular}{cc}
\includegraphics[height=2.3in,width=3in]{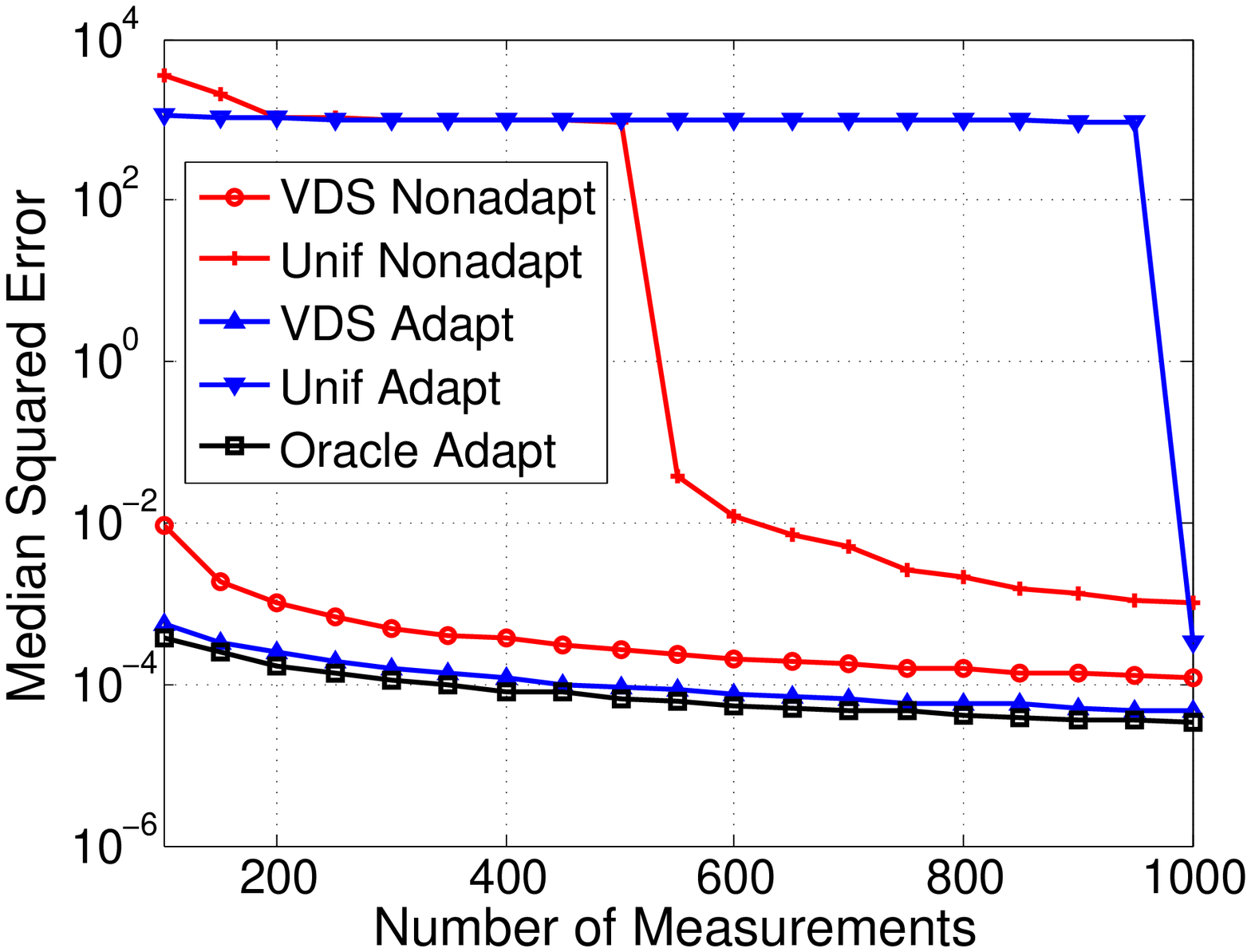} &
\includegraphics[height=2.3in,width=3in]{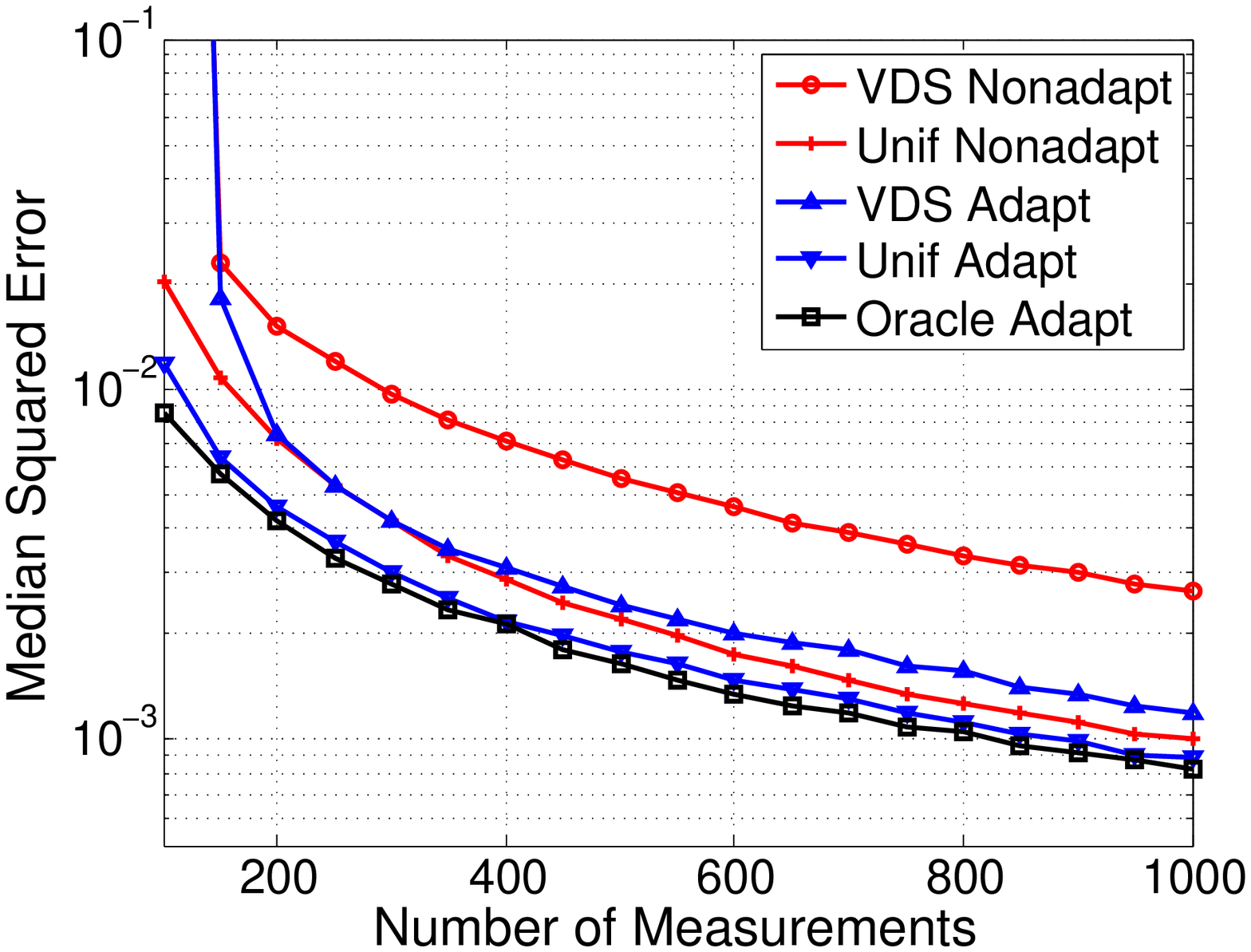} \\
\includegraphics[height=2.3in,width=3in]{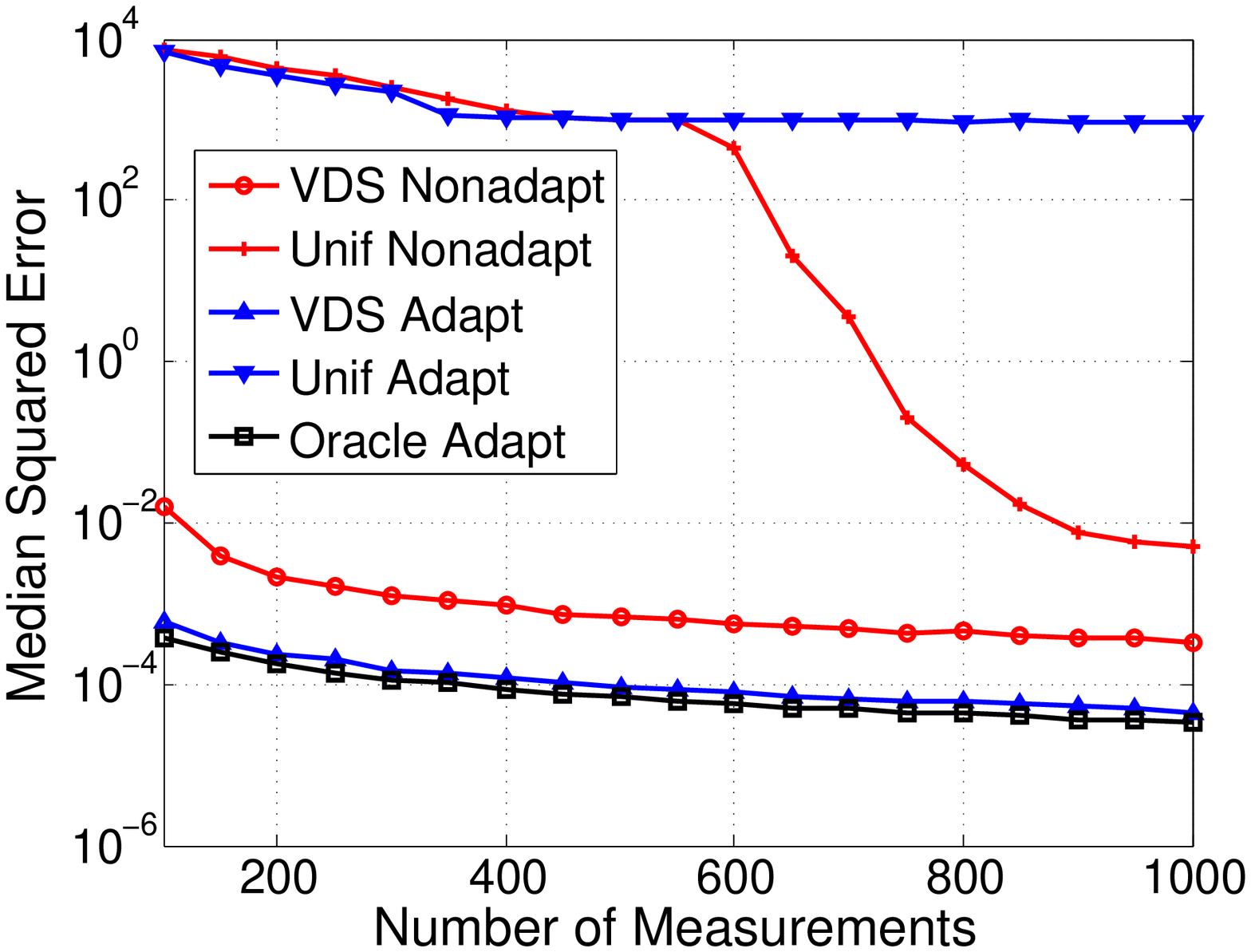} &
\includegraphics[height=2.3in,width=3in]{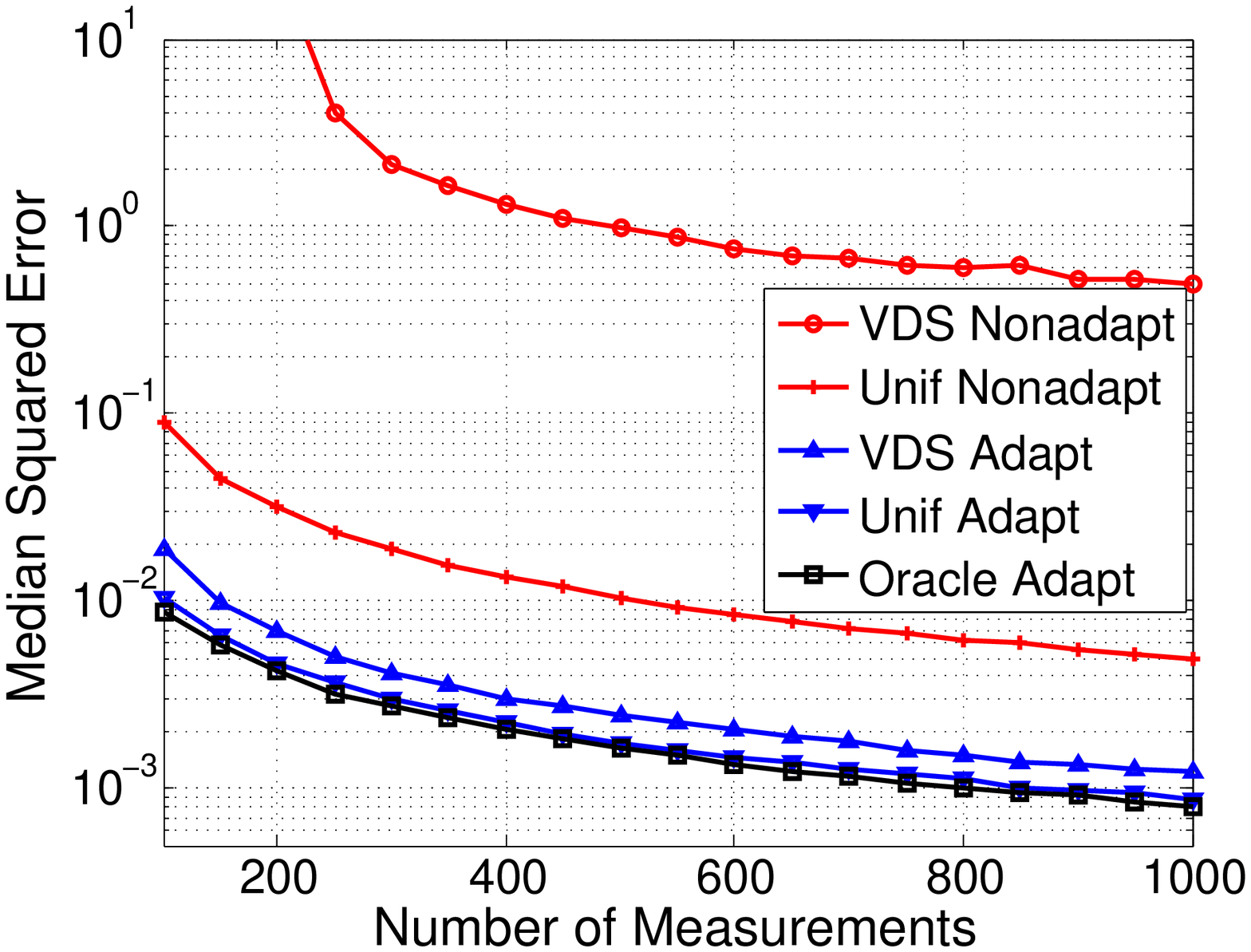}
\end{tabular}
\caption{\small{\editnew{(Large measurement regime)} The median squared error versus the number of measurements $m$ when the nonzero locations of $\boldsymbol{\alpha}$ are selected on a sparse tree (\editnew{left}) or uniformly at random (\editnew{right}). The nonadaptive (red) and adaptive (blue) recovery is shown when either VDS or uniform sampling is used for the nonadaptive measurements, and CoSaMP (\editnew{top}) or $\ell_1$-minimization (\editnew{bottom}) is used; the oracle adaptive (black) recovery is also included for comparison.  
}}
\label{fig::MSEvsMeas}
\end{figure*}

\begin{figure*}[t]
\centering
\begin{tabular}{cc}
\includegraphics[height=2.3in,width=3in]{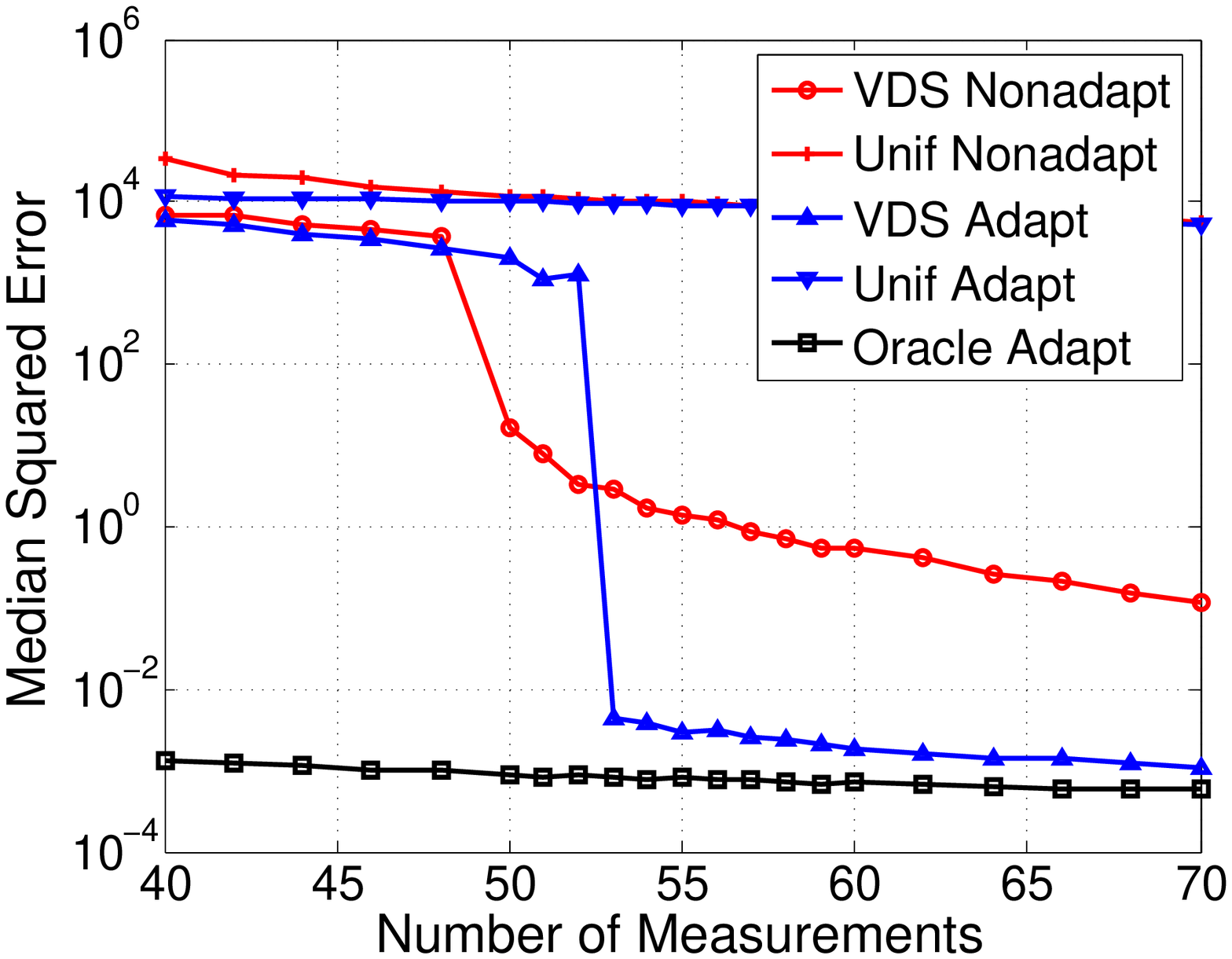} &
\includegraphics[height=2.3in,width=3in]{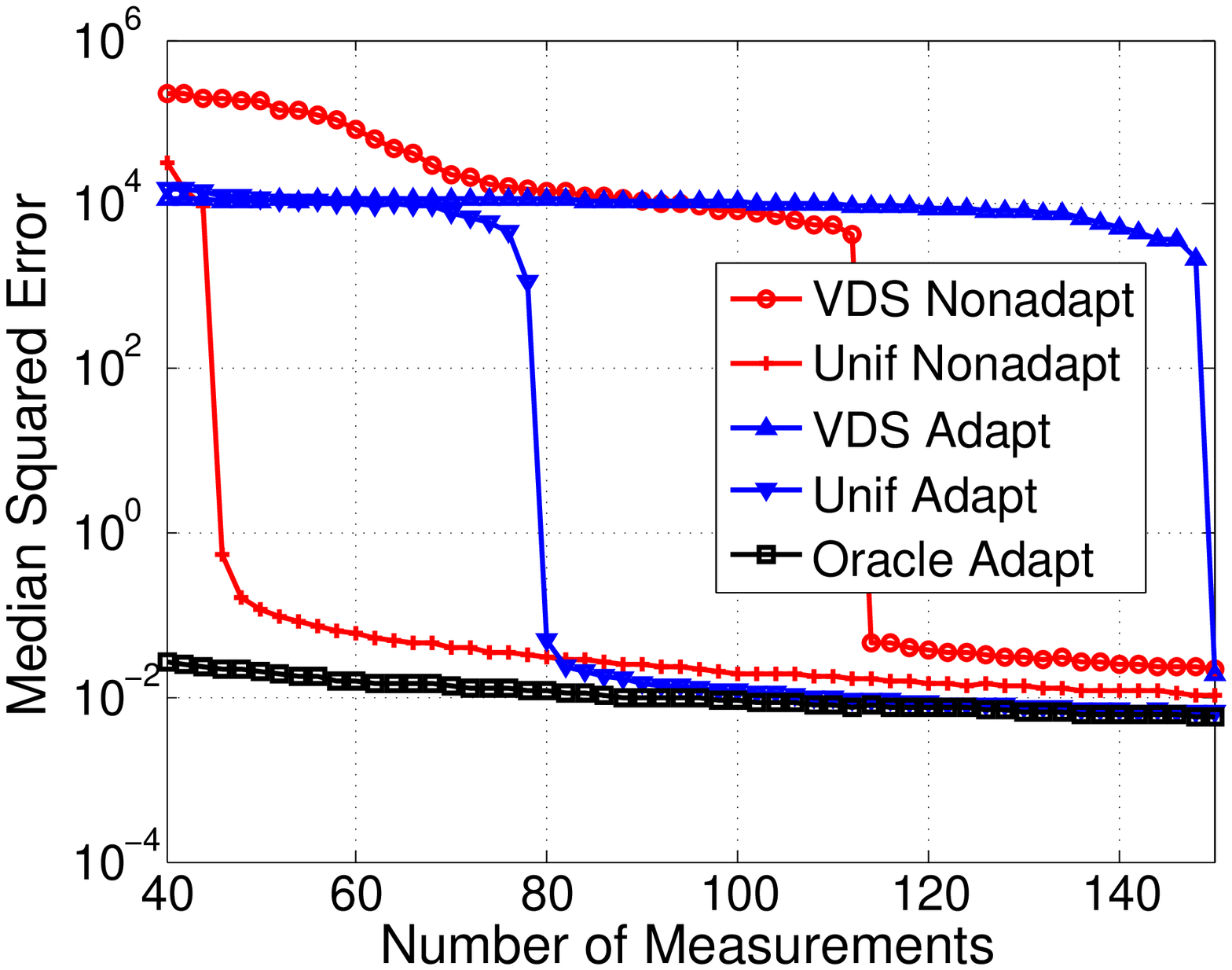}
\end{tabular}
\caption{\small{\editnew{(Small measurement regime) The median squared error versus the number of measurements $m$ when the nonzero locations of $\boldsymbol{\alpha}$ are selected on a sparse tree (left) or uniformly at random (right). The nonadaptive (red) and adaptive (blue) recovery is shown when either VDS or uniform sampling is used for the nonadaptive measurements, and CoSaMP is used; the oracle adaptive (black) recovery is also included for comparison.}
}}
\label{fig::MSEvsMeas_small}
\end{figure*}

\begin{figure*}[t]
\centering
\begin{tabular}{cc}
\includegraphics[height=2.3in,width=3in]{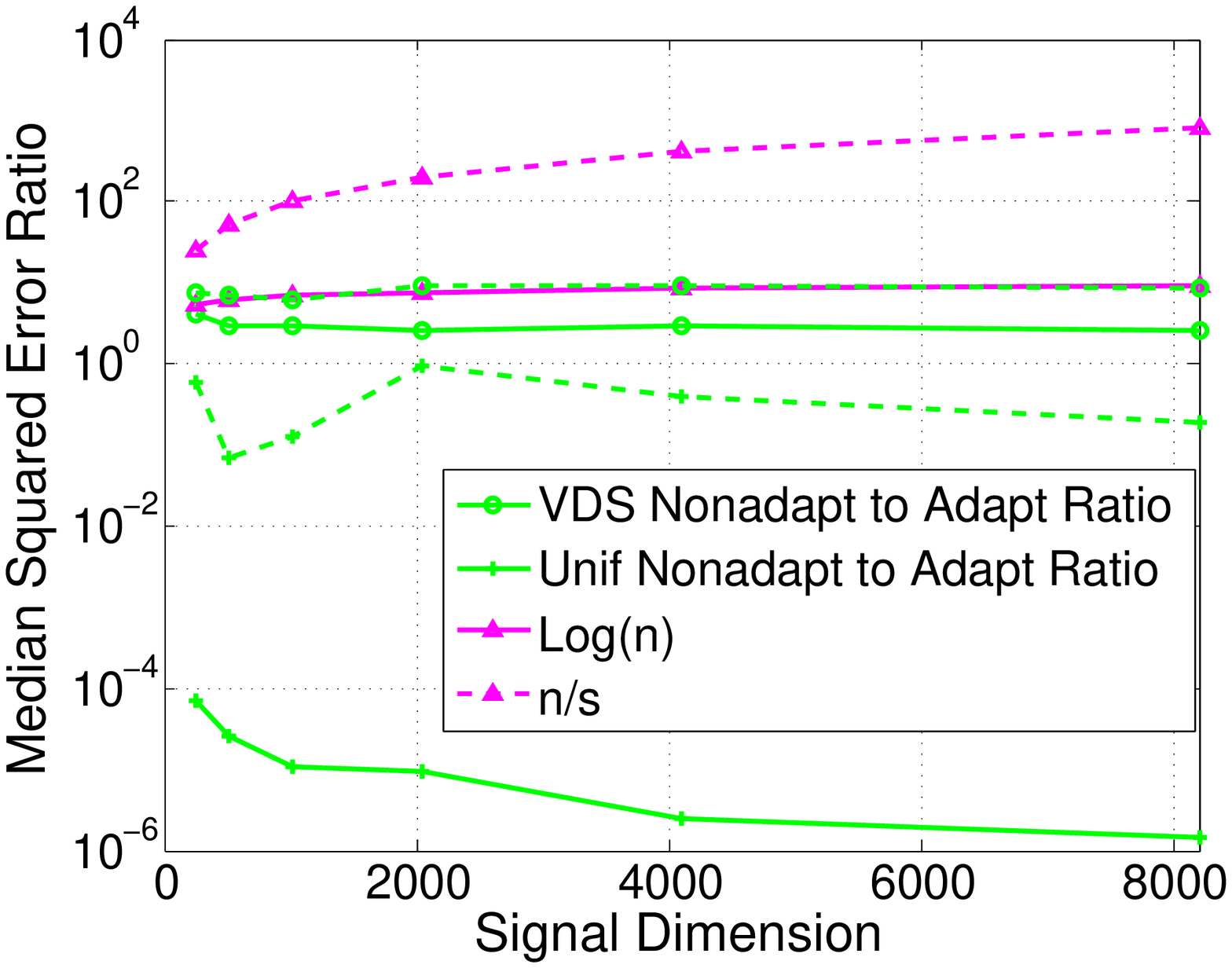} &
\includegraphics[height=2.3in,width=3in]{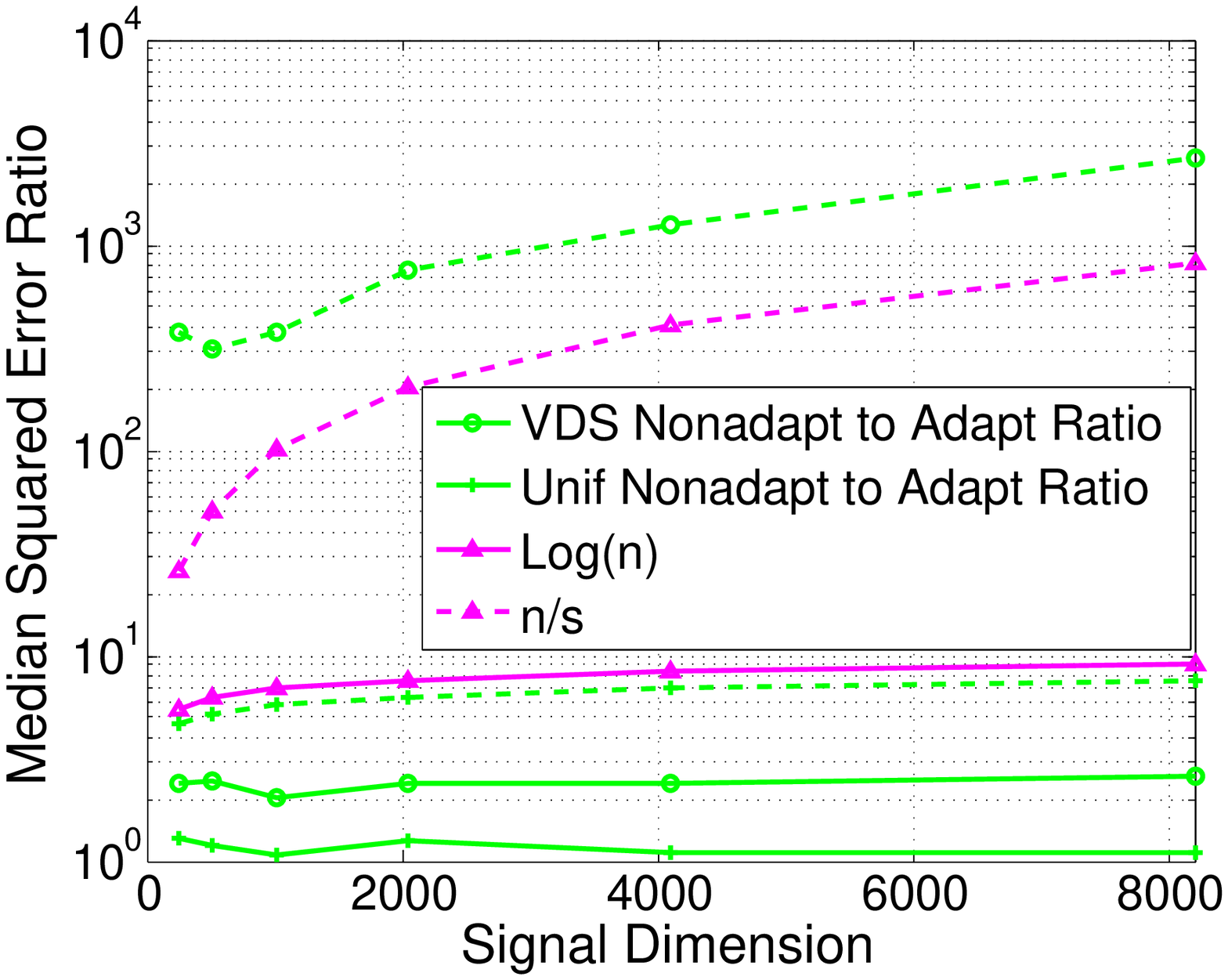}
\end{tabular}
\caption{ \small{The ratio (green) of the nonadaptive median squared recovery error to the adaptive median squared recovery error versus the signal dimension $n$ when the support locations of $\boldsymbol{\alpha}$ are selected on a sparse tree (left) or uniformly at random (right). The ratio is shown when either VDS or uniform sampling is used for the nonadaptive measurements, and CoSaMP (solid line) or $\ell_1$-minimization (dashed line) is used. The curves for $\log n$ (solid magenta) and $n/s$ (dashed magenta) are included for comparison. }}
\label{fig::MSEvsDim}
\end{figure*}

\subsection{Simulations} \label{subsec:simulations}

Here we present a practical implementation of adaptive sensing obtained via the relaxation (\ref{eqn::relaxation}) which we then compare with the results of traditional nonadaptive sensing. 
To implement (\ref{eqn::relaxation}), we use the Templates for First-Order Conic Solvers (TFOCS) software package \cite{BeckeCG_Templates,BeckeCG_TemplatesGuide}. 

For our first two sets of experiments, we set $\mathcal{M}$ to be the ensemble of $n$ measurements from the $n\times
n$ DFT matrix $\F$. We define $\x$ to be a 10-sparse
signal in the Haar wavelet basis (i.e., $\mathbf{\Psi} = \H^\star$) with the
values on the support of $\boldsymbol{\alpha}$ distributed i.i.d as $\oper
N(\sqrt{n},1)$ and the measurement noise $\z$ is distributed as i.i.d.
$\oper N(0,10^{-4})$.\footnote{We have found the adaptive procedure to be robust to the noise level, and compare similarly to the corresponding nonadaptive procedure even for larger noise levels.}  Unless otherwise stated, the signal is of dimension
$n=1024$.  We consider signals whose support is chosen uniformly at random,
and also those whose support obeys a tree structure. 
Briefly, in the latter
case the support is organized on a binary tree, plus an extra node at the top.
The first scaling (or lowest frequency) coefficient has just one child; the
second and further wavelet coefficients have two children each.
This model is characteristic of natural images which tend to have
inter-scale correlations (see
\cite{DuartWB_Fast,CrousNB_Wavelet} for similar wavelet-tree constructions).
An $s$-sparse support is filled by choosing the first scaling location, and then in each of the $s-1$ remaining rounds, choosing one node randomly among the unfilled nodes which currently have a chosen parent.

{\bfseries Nonadaptive sensing.} 
Due to the lack of incoherence between the DFT and Haar bases, it has been observed (and recently theoretically shown \cite{KrahmW_VDS,KrahmerCompressive15}) that so-called {\it Variable-Density Sampling} (VDS) is often preferable to standard uniform random selection of DFT measurements.  In VDS\footnote{Following the experiments in \cite{KrahmW_VDS}, we also do not apply any preconditioning to the sensing matrix.}, sampling can be concentrated on the lower frequencies, producing superior recovery results. We test recovery using either $\ell_1$-minimization \cite{CandeRT_Stable,BergF_Probing,spgl1:2007} or the greedy pursuit CoSaMP \cite{NeedeT_CoSaMP}.

{\bfseries Adaptive sensing.}  
In the more realistic setting, we employ a simple strategy which uses $m/2$ nonadaptive measurements (using either VDS or uniform sampling) to construct an estimate of $\Lambda$.  This is done by executing either $\ell_1$-minimization (followed by thresholding) or CoSaMP. 
We then solve the relaxation (\ref{eqn::relaxation}) using this estimated
support, and the remaining $m/2$ measurements are selected adaptively\footnote{Note that these $m/2$ {\it adaptively} selected measurements are only adapted to the first $m/2$ measurements, but are nonadaptive with respect to each other. That is, only one instance of adaptive measurement selection is being performed. Although in a different context, a similar two-stage approach is also taken in \cite{CastrW_Faster}.} using the distribution given by~\eqref{eq:Spmf}.  To recover the signal, either $\ell_1$-minimization or CoSaMP
is again used to obtain an updated estimate $\hat{\Lambda}$ using all $m$
measurements.\editnew{\footnote{\editnew{Using dependent measurements is of course not justified theoretically, but we found unsurprisingly that using all $m$ measurements gave better empirical results.}}} The final signal coefficient estimate is calculated as $\widehat{{\boldsymbol \alpha}}_{\hat{\Lambda}} = (\F'\H_{\hat{\Lambda}}^*)^\dagger \y$, where $\y$ is
the $m$-dimensional vector of measurements, $\F'$ is the $m\times n$ vector of
DFT measurements selected, and $\H_{\hat{\Lambda}}^*$ is the $n\times 10$
submatrix of $\H^*$ restricted to the columns indexed by $\hat{\Lambda}$. One
could alternatively use the signal estimate returned directly from the recovery algorithm, which we have observed to perform similarly to (or only slightly worse than) our implemented method.

{\bfseries Oracle adaptive sensing.} For sake of comparison, we also consider the case where the true support $\Lambda$ of the signal is known a priori, and the measurements are selected as in the adaptive sensing case using this $\Lambda$. 
Recovery is then performed simply by applying the pseudoinverse: $\widehat{{\boldsymbol \alpha}}_{\Lambda} = (\F'\H_{\Lambda}^*)^\dagger \y$.

Figure \ref{fig::MSEvsMeas} 
compares recovery results over 1000 trials for nonadaptive, adaptive, and oracle adaptive sensing versus the number of measurements $m$\editnew{, where $m$ ranges between 100 and 1000}. 
We see that when the signal is supported on a tree, uniform sampling performs poorly for both nonadaptive and adaptive sensing, as might be expected. The performance of the uniform sampling methods can be understood via the empirical observation that in this case we require roughly $500$ measurements before we can reliably estimate the support.  When using the CoSaMP algorithm, the sudden improvement at $m \approx 500$ for uniform nonadaptive and at $m \approx 1000$ for uniform adaptive (which uses $m \approx 500$ measurements for support identification) corresponds to the threshold where more than half of the trials resulted in a correct support recovery.  In contrast, sampling with VDS offers dramatic improvements for both nonadaptive and adaptive sensing with either reconstruction algorithm, with adaptive sensing performing almost as well as the oracle.  In this case, VDS is already capturing much of the potential improvement offered by adaptivity because the energy of the signal is heavily biased towards the lower frequencies, although adaptivity still results in somewhat improved performance.   
In contrast to the tree-sparse case, when the signal support is selected randomly, uniform nonadaptive sampling actually performs better than VDS, whereas adaptive sensing performs similarly regardless of the type of nonadaptive measurements taken.  Thus if one is not sure of the signal structure in general, adaptive sensing can offer improvements in either case.  This flexibility represents one of the main advantages of adaptive sensing.

\editnew{Figure \ref{fig::MSEvsMeas_small} studies the same setting as Figure \ref{fig::MSEvsMeas} when using CoSaMP for recovery, but focuses on the small measurement regime. These results illustrate that there are regions, however narrow, where the nonadaptive method can succeed while the adaptive method fails. This is expected due to the nature of the adaptive scheme, where only $m/2$ measurements are utilized to identify a support estimate. At some point, the support can be sufficiently estimated with $m$, but not $m/2$, measurements. For tree-sparse signals, nonadaptive sensing with VDS measurements outperforms adaptive sensing with VDS nonadaptive measurements when $m\approx 50$. For uniformly sparse signals, we see this behavior even more clearly for both VDS and uniform sampling.}

The results of our second simulation are shown in Figure \ref{fig::MSEvsDim}, where we compare the ratio of nonadaptive to adaptive sensing recovery over 200 trials against the dimension $n$ of the signal $\x$; the number of measurements used is always $m=0.6n$ \editnew{(rounding when necessary).} We note that since the norms of $\boldsymbol{\alpha}$ and $\z$ both scale with $n$, the SNR remains roughly the same for all signal dimensions $n$.  We observe similar results as Figure~\ref{fig::MSEvsMeas}, demonstrating the behavior holds as a function of dimension.

In our last experiment, we evaluate our adaptive approach on real images.  This scenario differs from previous experiments in two key aspects.
First, the signal of interest is a two-dimensional (2D) image, not a one-dimensional
vector, and thus we use 2D DFT measurements and a
2D discrete wavelet transform as the sparsity basis.
Second, the image is not exactly sparse in any wavelet basis.
Hence, when estimating the sparse support we introduce an additional
(non-Gaussian) source of error, the contribution of the off-support wavelet
coefficients. \editnew{We note that the choice of the parameter $s$, which we have not attempted to optimize,
can have an impact on signal reconstruction.}

The image we use, {\small\texttt{brain.mat}}\footnote{Obtained from
\url{http://www.eecs.berkeley.edu/~mlustig/CS.html}.},
is rescaled to be
$64\times 64$, and is shown in Figure~\ref{fig:mri-comparison}.
We use the Daubechies wavelet with 3 vanishing moments (D6)
in a full 2D decomposition (i.e., $\log_2 64=6$ levels).
\editnew{We set the parameter $s=1000$, which we again note was not tuned nor optimized.}
Additionally, we introduce white Gaussian noise at the level of $\sigma=0.01$
to each measurement. 

The experiment proceeds as follows: in the nonadaptive case, $m$
measurements are taken according to VDS.  The set of recovered wavelet
coefficients are obtained using $\ell_1$-minimization and the image is
reconstructed using the inverse wavelet transform.  Note that the output of $\ell_1$-minimization
is not necessarily exactly $s$-sparse.  The assumed sparsity $s$
guides our choice of the $\ell_2$ error term constraint, but we did no
thresholding afterwards.  We evaluate performance by the median peak
signal-to-noise ratio (PSNR) in dB over 50 trials.

In the adaptive case, $m/2$ VDS measurements are taken as in the previous
nonadaptive case. Then, via the $\ell_1$-minimization reconstruction,
we determine the estimated top $s$ wavelet coefficients in each of 50
trials.  We choose the trial with accuracy (in terms of the number of correctly
identified top $s$ wavelet coefficients) closest to the \editnew{median} accuracy.  Utilizing
the size $s$ support estimate identified, we solve the relaxation (\ref{eqn::relaxation}) and select
the remaining $m/2$ measurements adaptively. Finally, we recover the signal
via $\ell_1$-minimization, and, as before, reconstruct the final wavelet coefficients
using the pseudoinverse. Again, we
evaluate performance by the median PSNR over 50 trials of the adaptive measurement selection.

The adaptive and nonadaptive recovered images of the single trial with the
closest to median PSNR performance using a total of $m=3000$ measurements are given in
Figure~\ref{fig:mri-comparison}. Notice that the PSNR of the adaptive strategy is 28.02 dB, which
exceeds the PSNR of 25.03 dB of the nonadaptive strategy. Visually, the adaptive strategy more
closely resembles the original image.
The median PSNR as the number of measurements
is varied is shown in Figure~\ref{fig:mri-median-plot}.  The plot shows that as
the number of measurements reaches a certain level (roughly above 2000 measurements), the two-stage adaptive
approach begins to exceed the method which is purely nonadaptive. Hence, as long as enough nonadaptive measurements are taken to obtain a sufficient support estimate, the adaptive procedure can improve image reconstruction quality.

We note that adaptive approaches to medical imaging have also been
studied using an alternative Bayesian model for sampling optimization
\cite{SeegerN_OptKspace,NickiP_Bayesian,Seeger_2009speeding}. The work
\cite{Seeger_2009speeding}
studies the optimization of sequential sampling over stacks of neighboring image slices. In future work, it would be interesting to extend our proposed adaptive sampling scheme to this setting. Our method, however, is a framework for general adaptive sensing, not tuned specifically for medical imaging.

\begin{figure}[tb]
\centering
\includegraphics[height=1.3in,width=1.3in]{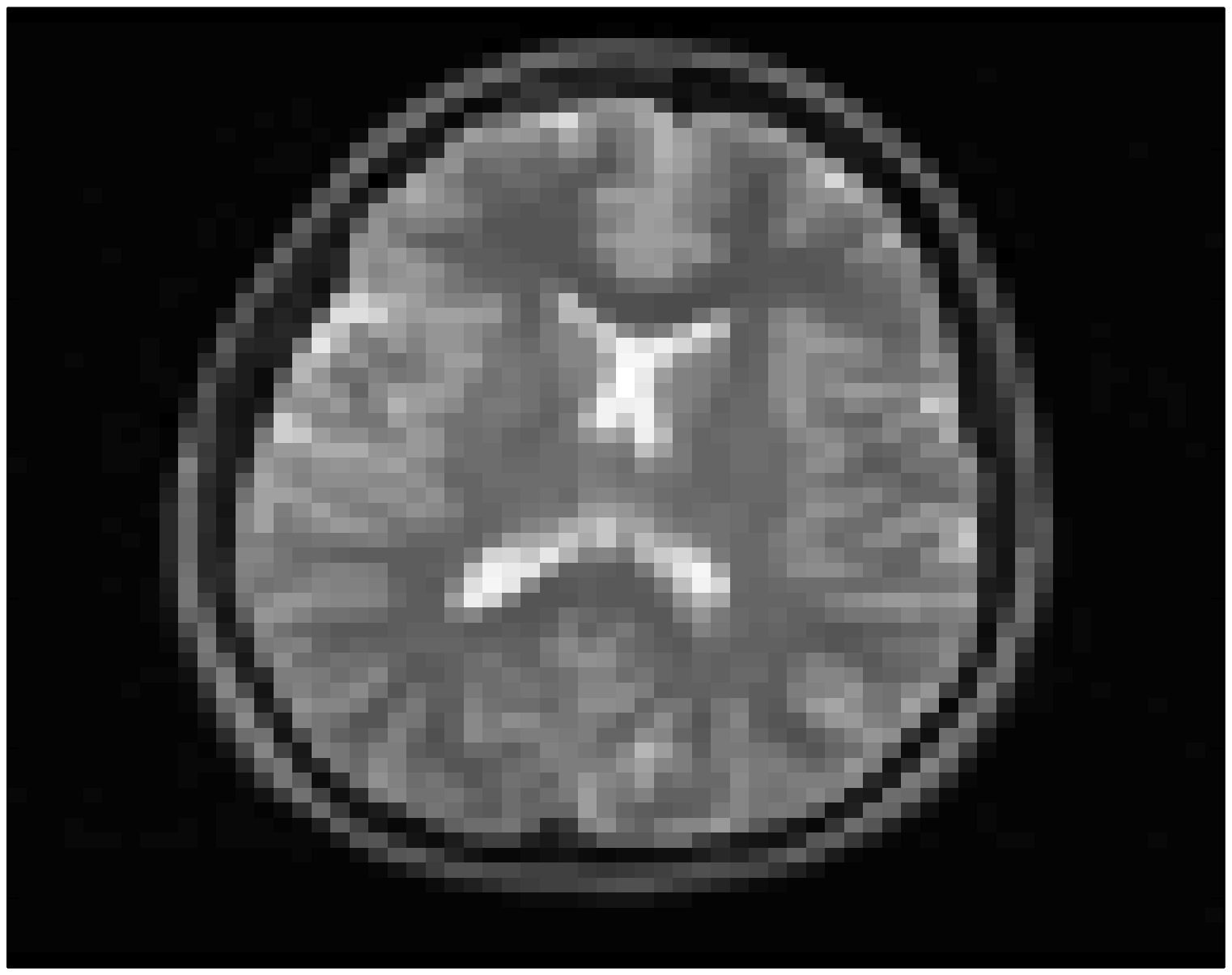} \\
\small{Original}\\
\vspace{2mm}
\begin{tabular}{cc}
\includegraphics[height=1.3in,width=1.3in]{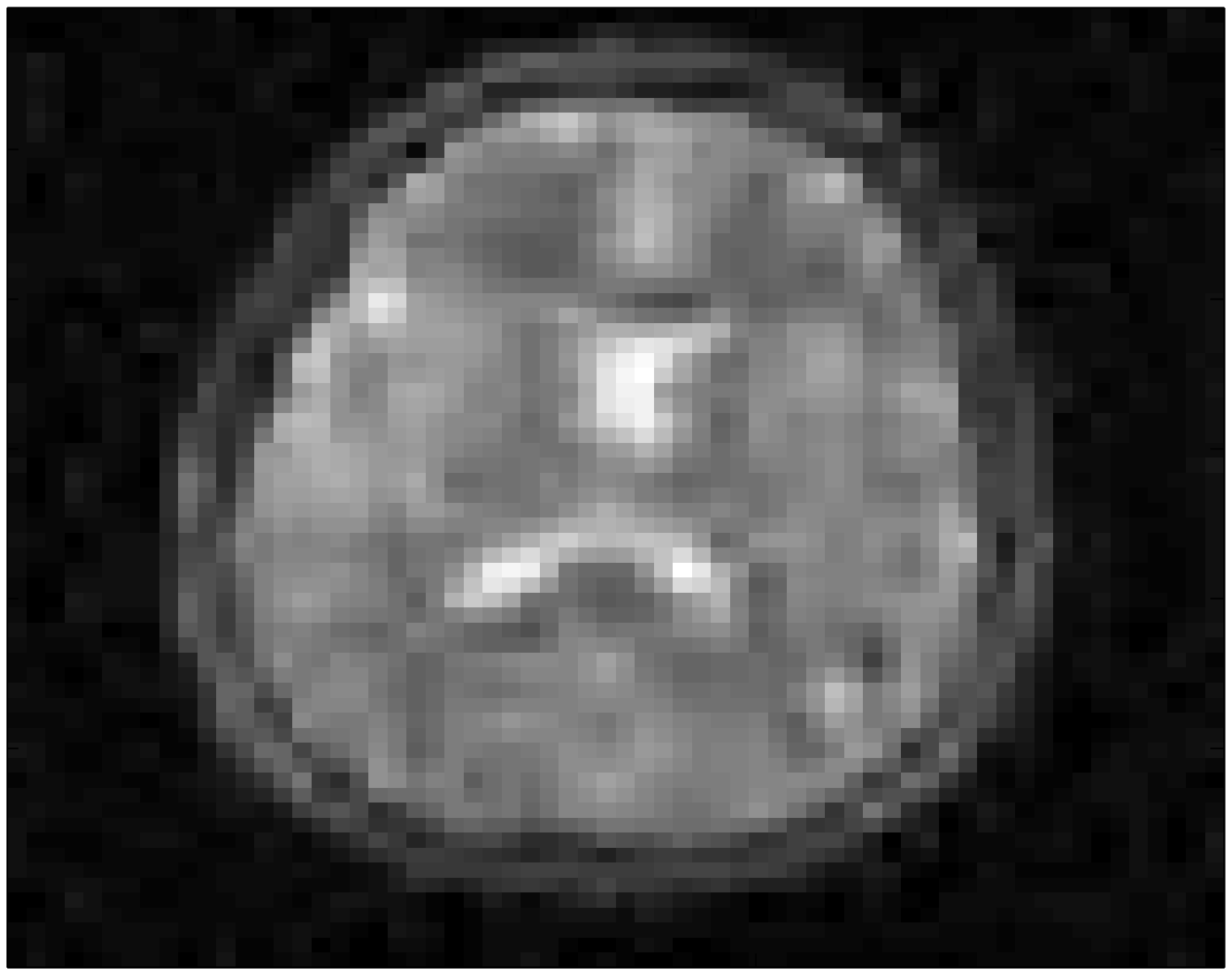}
&
\includegraphics[height=1.3in,width=1.3in]{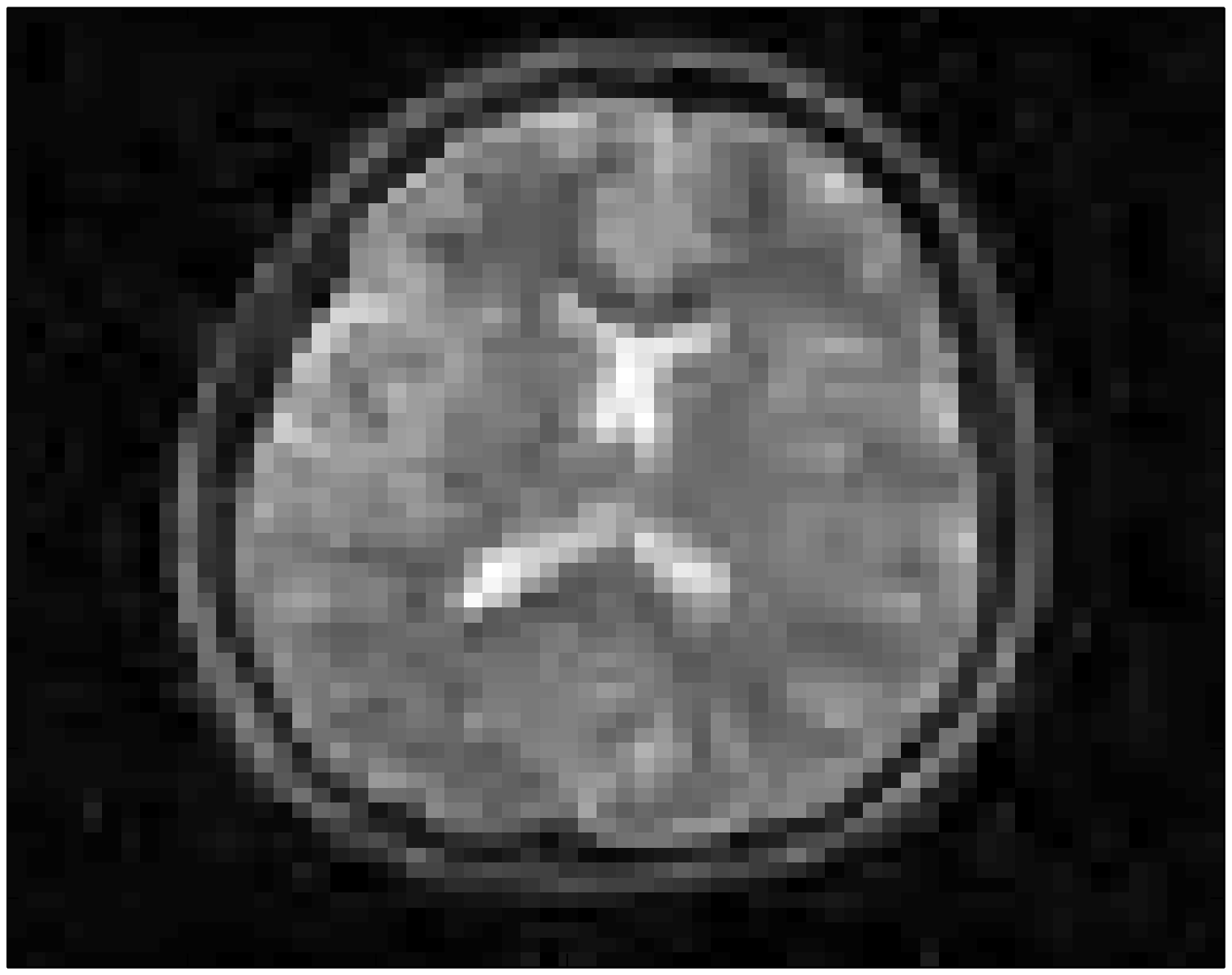}\\
Nonadaptive,& Adaptive,\\
PSNR $=25.03\,$dB.& PSNR $=28.02\,$dB.
\end{tabular}
\caption{(Top) The $64\times 64$ \texttt{brain.mat} image used in our medical
imaging experiments.  Reconstructed images with the closest to median
PSNR among the 50 trials of nonadaptive (bottom left) and adaptive (bottom right) sensing
with $m=3000$ measurements.}
\label{fig:mri-comparison}
\end{figure}

\begin{figure}[tb]\centering
\includegraphics[height=2.3in,width=3in]{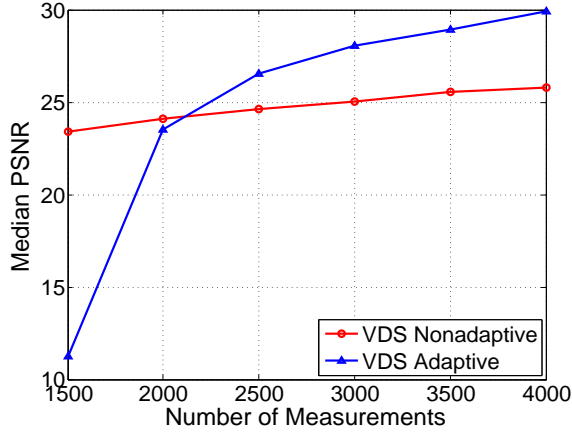}
\caption{The median PSNR versus the number of measurements $m$ over 50 trials
of nonadaptive and adaptive sensing of the $64\times 64$ \texttt{brain.mat}
image.}
\label{fig:mri-median-plot}
\end{figure}

\subsection{Analysis of the 1-sparse case}

We now provide some analytical justification explaining why adaptive sensing can achieve a lower MSE than nonadaptive sensing for the Haar wavelet basis with DFT measurements, but show that the largest gains are realized for a small fraction of the possible signal support sets.  We consider the simple case when $s=1$ and the support is eventually known (either by oracle or by utilizing some method for estimation, as in the above experiments), and use this toy problem as motivational justification for the general setting.  
If we denote the $s$ singular values of $\F'\H_\Lambda^*$ by $\sigma_1\geq \sigma_2 \geq \dots \geq \sigma_s > 0$, then in general we want to minimize
\begin{equation} \label{eqn::sumSingVal}
\| (\F'\H_\Lambda^*)^\dagger \|_F^2  = \sum_{i=1}^s \frac{1}{\sigma_i^2}.
\end{equation}
When $s=1$, (\ref{eqn::sumSingVal}) becomes $\| (\F'\H_\Lambda^*)^\dagger \|_F^2  = \frac{1}{\sigma_1^2}$. However, minimizing this quantity is the same as maximizing $\| \F'\H_\Lambda^* \|_F^2  = \sigma_1^2$. It is easy to see that $\| \F'\H_\Lambda^* \|_F^2$ is maximized when a measurement of $\F$ that is most correlated with $\H_\Lambda^*$ is chosen for every measurement in $\F'$. That is, if such a row can be identified, the best way to sense 1-sparse signals adaptively once the support is known is to simply repeat that measurement until the number of allotted measurements has been reached. Note that $\F'\H_\Lambda^*$ is $m\times 1$ in this case, and is still full rank when selecting the measurements in this way; thus, the theory leading to (\ref{eqn::MSEandNormFH}) still holds. In this setting, we can determine explicitly what the MSE looks like, and provide bounds on the MSE that depend on the support location of the 1-sparse signal. 
The result assuming a known support is provided in Theorem \ref{thm::1sparseMSE}. The result in the more realistic context of adaptive sensing, where the first half of the measurements are selected nonadaptively, immediately follows and is provided in Corollary \ref{cor::1sparseAdap}.

\begin{theorem} \label{thm::1sparseMSE}
Denote by $\x=\H^*\boldsymbol{\alpha}$ the signal of interest, and suppose $\x$ becomes 1-sparse after applying the Haar wavelet transformation $\H$ (that is, $\H\x = \boldsymbol{\alpha}$ and $\| \boldsymbol{\alpha} \|_0 = 1$). Let $\supp(\boldsymbol{\alpha}) = \Lambda$, and suppose the support $\Lambda$ is completely known. Suppose we measure repeatedly with a particular measurement from the $n\times n$ DFT $\F$ defined in (\ref{dft}); denote this measurement by $\mathbf{f}_j$, where $j\in\{0,1,\dots,n-1\}$ is some row index. Then, our observations are of the form
\begin{equation} \label{eqn::meas}
y_i = \inner{\mathbf{f}_j,\H^*_\Lambda\boldsymbol{\alpha}_\Lambda} + z_i,
\end{equation}
for $i = 1,\dots,m$, where the noise $z_i$ are i.i.d. $N(0,\sigma^2)$. Then the MSE
is given by
\begin{equation} \label{eqn::FH MSE thm oracle}
\mathbb{E}\|\widehat{\x}-\x \|_2^2 =   \frac{1/m}{|\inner{\mathbf{f}_j,\H^*_{\Lambda}}|^2}\sigma^2,
\end{equation}
and is bounded by
\begin{equation} \label{eqn::error bound oracle}
\frac{\sigma^2}{m} \leq \mathbb{E}\|\widehat{\x}-\x \|_2^2 \leq \frac{n\sigma^2}{2m},
\end{equation}
where the expectation is taken with respect to $\z$.
\end{theorem}

Note that in standard compressive sensing when we rely on the RIP, the DFT matrix $\F$ is normalized by $\frac{1}{\sqrt{m}}$ rather than $\frac{1}{\sqrt{n}}$. If we make this normalization, the bound in (\ref{eqn::error bound oracle}) becomes
\begin{align*}
\frac{\sigma^2}{n} &\leq \mathbb{E}\|\widehat{\x}-\x\|_2^2 \leq \frac{\sigma^2}{2}.
\end{align*}
Including pre-conditioning and other scalings of course yields an analogous bound.

The proof of Theorem \ref{thm::1sparseMSE} relies on three lemmas that provide bounds for the term $|\inner{\mathbf{f}_j,\H^*_{\Lambda}}|$ appearing in the MSE in (\ref{eqn::FH MSE thm oracle}). Specifically, one term of interest is
\begin{align}
\label{lowerbound}
\min_{\Lambda\in\{0,\dots,n-1\}} \max_{j\in\{0,\dots,n-1\}} | \inner{\mathbf{f}_j, \H^*_{\Lambda}} |.
\end{align}
The maximization over $j$ corresponds to selecting the best DFT measurement $\mathbf{f}_j$, and the minimization accounts for the worst case signal (i.e., the worst case support $\Lambda$). On the other hand, we also want to obtain a value that represents the best case signal 
so we are also interested in
\begin{align}
\label{upperbound}
\max_{\Lambda\in\{0,\dots,n-1\}} \max_{j\in\{0,\dots,n-1\}} | \inner{\mathbf{f}_j, \H^*_{\Lambda}} |.
\end{align}

Before proving Theorem \ref{thm::1sparseMSE}, let us set some notation. The Haar wavelet transform matrix $\H$, defined in (\ref{eqn::HaarDef1}) and (\ref{eqn::HaarDef2}) consists of {\it blocks} of consecutive rows with the same nonzero entry magnitudes.  
Let $1 \leq a \leq\log_2n$ denote the block of $\H$, where $a=1$ corresponds to the $\frac{n}{2}$ rows indexed by $j = \frac{n}{2},\dots,n-1$ (i.e., the ``bottom" half of $\H$), $a=2$ corresponds to the $\frac{n}{4}$ rows indexed by $j= \frac{n}{2}-\frac{n}{4},\dots,\frac{n}{2}-1$, and so on. 
Similarly, for $\H^*$, instead of blocks of {\it rows}, we have blocks of {\it columns}; the block corresponding to $a=\log_2n$ represents the lowest frequency wavelets, and the block corresponding to $a = 1$ represents the highest frequency wavelets.

\begin{proof}[Proof of Theorem \ref{thm::1sparseMSE}]

This proof requires the following three lemmas. Lemmas \ref{lem::sumFourier} and \ref{lem::innerprod} are used to prove Lemma \ref{prop::nogoodadp}, and Lemma \ref{prop::nogoodadp} is used to complete the proof of Theorem \ref{thm::1sparseMSE}.  The lemmas can be derived using elementary trigonometric bounds, and we omit the proofs here. 

\begin{lem}
\label{lem::sumFourier}
Fix $j\in\mathbb{Z}$ where $1\leq j\leq n-1$ and let $a = 1,\dots,\log_2n$. Choose $k\in\mathbb{Z}$, $0\leq k\leq n-\frac{2^a}{2}$.  Then
\begin{align}
\left| \sum_{q=k}^{k+\frac{2^a}{2}-1} e^{-2\pi i j q/n}\right| &= \sqrt{\frac{1-\cos(\frac{2^a \pi j}{n})}{1-\cos(\frac{2\pi j}{n})}}.
\end{align}
\end{lem}

\begin{lem}
\label{lem::innerprod}
Let $\mathbf{f}_j$, $j\in\{0,\dots,n-1\}$, be row $j$ from the $n\times n$ DFT and let $\H^*_\Lambda$ be the inverse discrete Haar wavelet transform restricted to the column indexed by $\Lambda$. Let $a=1,\dots,\log_2n$ denote the block of $\H^*$ and let $\Lambda\in\{1,2,\dots,n-1\}$, $|\Lambda|=1$, be a column in the set corresponding to block $a$. Then,
\begin{align}
\label{eqn::innerprod}
|\inner{\mathbf{f}_j,\H_{\Lambda}^*}| &= \frac{1}{\sqrt{n2^{a-1}}}\frac{1-\cos(\frac{2^a \pi j}{n})}{\sqrt{1-\cos(\frac{2 \pi j}{n})}},
\end{align}
where
$j = 1,\dots,n-1$. When $j=0$,
\[ |\inner{\mathbf{f}_0,\H_\Lambda^*}| =
\begin{cases}
      1 & \Lambda = 0 \\
      0 & \Lambda \in \{1,\dots,n-1 \} .
   \end{cases}
\]
When $\Lambda=\{0\}$,
\[ |\inner{\mathbf{f}_j,\H_0^*}| =
\begin{cases}
      1 & j=0 \\
      0 & j=1,2,\dots,n-1 .
   \end{cases}
\]
\end{lem}

\begin{lem} \label{prop::nogoodadp}
Let $\mathbf{f}_j$, $j\in\{0,\dots,n-1\}$, be a row from the $n\times n$ DFT matrix $\F$ and let $\H^*_\Lambda$ be the inverse discrete Haar wavelet transform restricted to the column indexed by $\Lambda$. Let $\Lambda\in\{0,1,\dots,n-1\}$ so that $|\Lambda|=1$. Then
\begin{equation} \label{eqn::lowerboundFH}
 \min_{\Lambda\in\{0,\dots,n-1\}} \max_{j\in\{0,\dots,n-1\}} | \inner{\mathbf{f}_j, \H^*_{\Lambda}} | = \sqrt{\frac{2}{n}}
\end{equation}
and
\begin{equation}\label{eqn::upperboundFH}
\max_{\Lambda\in\{0,\dots,n-1\}} \max_{j\in\{0,\dots,n-1\}} | \inner{\mathbf{f}_j, \H^*_{\Lambda}} | = 1.
\end{equation}
\end{lem}

Since $|\Lambda| = 1$, $\boldsymbol{\alpha}_\Lambda$ is just a scalar, so that the measurements (\ref{eqn::meas}) can be written as
\begin{equation}
y_i = \inner{\mathbf{f}_j,\H^*_\Lambda}\boldsymbol{\alpha}_\Lambda + z_i.
\end{equation}
This can be concisely written as
\begin{align}
\y = \A \boldsymbol{\alpha}_\Lambda + \z,
\end{align}
where $\A$ is the $m$-dimensional column vector with each entry equal to $\inner{\mathbf{f}_j,\H^*_\Lambda}$. To estimate $\boldsymbol{\alpha}_\Lambda$, we apply $\A^\dagger$ to $\y$. In this case, $\A^\dagger$ is an $m$-dimensional row vector, with each entry equal to $\frac{1}{m\inner{\mathbf{f}_j,\H^*_\Lambda}}$. Therefore,

\begin{align*}
\hat{\boldsymbol{\alpha}}_\Lambda &= \A^\dagger \y
= \A^\dagger(\A \boldsymbol{\alpha}_\Lambda + \z)
\\&= \begin{bmatrix}
 \frac{1}{m\inner{\mathbf{f}_j,\H^*_\Lambda}} &\cdots & \frac{1}{m\inner{\mathbf{f}_j,\H^*_\Lambda}}
 \end{bmatrix}
 \left(\begin{bmatrix}
 \inner{\mathbf{f}_j,\H^*_\Lambda} \\
 \vdots \\
 \inner{\mathbf{f}_j,\H^*_\Lambda}
 \end{bmatrix} \boldsymbol{\alpha}_\Lambda
 +\z\right) \\
 &= \frac{1}{m} \sum_{i=1}^{m}\left(\boldsymbol{\alpha}_\Lambda + \frac{z_i}{\inner{\mathbf{f}_j,\H^*_\Lambda}} \right )
 = \boldsymbol{\alpha}_\Lambda + \frac{1}{m} \sum_{i=1}^{m} \frac{z_i}{\inner{\mathbf{f}_j,\H^*_\Lambda}}.
\end{align*}
Using this, and 
since $\sum_{i=1}^{m}z_i\sim\oper N(0,m\sigma^2)$, we find
\begin{align*}
\mathbb{E}\|\widehat{\x}-\x\|_2^2 &= \mathbb{E}\|\H^*(\widehat{\boldsymbol{\alpha}}-\boldsymbol{\alpha})\|_2^2
\\&= \mathbb{E}\|\widehat{\boldsymbol{\alpha}}-\boldsymbol{\alpha}\|_2^2
= \mathbb{E}|\widehat{\boldsymbol{\alpha}}_{\Lambda} - \boldsymbol{\alpha}_{\Lambda}|^2 \\
&= \mathbb{E}\left|\boldsymbol{\alpha}_\Lambda -\left(\boldsymbol{\alpha}_\Lambda +  \frac{1}{m}\sum_{i=1}^{m} \frac{z_i}{\inner{\mathbf{f}_j,\H^*_\Lambda}}\right)\right|^2
\\&= \mathbb{E}\left|\frac{1}{m}\sum_{i=1}^{m}\frac{z_i}{\inner{\mathbf{f}_j,\H^*_{\Lambda}}}\right|^2
\\&= \left(\frac{1/m}{|\inner{\mathbf{f}_j,\H^*_{\Lambda}}|}\right)^2 \mathbb{E}\left|\sum_{i=1}^{m}z_i\right|^2
\\& = \left(\frac{1/m}{|\inner{\mathbf{f}_j,\H^*_{\Lambda}}|}\right)^2 m\sigma^2
= \frac{1/m}{|\inner{\mathbf{f}_j,\H^*_{\Lambda}}|^2}\sigma^2. \label{eqn::FH MSE} \numberthis
\end{align*}

Applying the bounds from Lemma \ref{prop::nogoodadp} to (\ref{eqn::FH MSE}), we arrive at
$$  \frac{\sigma^2}{m} \leq \mathbb{E}\|\widehat{\x}-\x\|_2^2 \leq \frac{n\sigma^2}{2m}, $$
as desired.
\end{proof}

\begin{cor}\label{cor::1sparseAdap}
Suppose $\x = \H^\star\boldsymbol{\alpha}$ is 1-sparse. Suppose after $\frac{m}{2}$ nonadaptive DFT measurements, the support $\Lambda$ is correctly identified. For the remaining $\frac{m}{2}$ DFT measurements, we measure repeatedly with a particular measurement from the $n\times n$ DFT $\F$; denote this measurement by $\mathbf{f}_j$, where $j\in\{0,1,\dots,n-1\}$ is some row index. Then our observations are of the form (\ref{eqn::meas}) for $i = \frac{m}{2}+1,\dots,m$, where the noise $z_i$ are i.i.d. $N(0,\sigma^2)$. Then the MSE 
is given by
\begin{equation} \label{eqn::FH MSE thm}
\mathbb{E}\|\widehat{\x}-\x \|_2^2 =   \frac{2/m}{|\inner{\mathbf{f}_j,\H^*_{\Lambda}}|^2}\sigma^2,
\end{equation}
and is bounded by
\begin{equation} \label{eqn::error bound}
\frac{2\sigma^2}{m} \leq \mathbb{E}\|\widehat{\x}-\x \|_2^2 \leq \frac{n\sigma^2}{m},
\end{equation}
where the expectation is taken with respect to $\z$.
\end{cor}

The upper bound on $\mathbb{E}\|\widehat{\x}-\x\|_2^2$ in Corollary \ref{cor::1sparseAdap} is precisely the lower bound from Theorem \ref{thm1} when $s=1$. This means that there is indeed some room for improvement with adaptive sensing when the sparsity basis is the Haar wavelet transform rather than the canonical basis.
Corollary \ref{cor::1sparseAdap} shows that the performance of adaptive sensing, in terms of the MSE, depends on the support location of 1-sparse signals. 
The best adaptive recovery is possible when the support is located on the lowest wavelet frequency ($\Lambda = \{0\}$, or the first Haar wavelet coefficient) while the worst recovery occurs when the support is located on any of the higher wavelet frequencies in block $a=1$ (the latter half of the Haar wavelet coefficients).
This of course matches the intuition based on the correlations in these two
bases. This suggests that structured signals such as those that are
tree-sparse will benefit more from adaptivity than signals that have
a uniformly distributed support.

In light of the discrepancy between \eqref{eqn::lowerboundFH} and \eqref{eqn::upperboundFH}, one wishes to know in some sense, what fraction of signals allow for recovery more like one versus the other.
Figure \ref{fig::ubfhinv} shows how $\max_{j\in\{0,\dots,n-1\}}| \inner{\mathbf{f}_j, \H^*_{\Lambda}} |$ varies by maximizing $\max_{j\in\{0,\dots,n-1\}}| \inner{\mathbf{f}_j, \H^*_{\Lambda}} |$ over $\Lambda$ while successively removing blocks from $\H^*$. Using our notation for blocks, the blocks of $\H^*$ are removed in the following (top-down) order: $\log_2(n),\log_2(n)-1,\dots,1$.  Then, we plot the value of $\max_{j\in\{0,\dots,n-1\}}| \inner{\mathbf{f}_j, \H^*_{\Lambda}} |$ for the remaining submatrix of $\H^*$.  
Hence, we see that the MSE is higher for signals supported on higher wavelet frequencies, and the upper bound of (\ref{eqn::error bound}) is achieved by exactly half of the possible signal support sets, whereas the lower bound of (\ref{eqn::error bound}) is achieved by exactly one of the possible signal supports sets (i.e., $\Lambda = \{0\}$).
Fortunately, the support of natural images tends to be
concentrated on lower-frequency
wavelet coefficients~\cite{petrosian2001wavelets}.

\begin{figure}[tb]
\centering
\includegraphics[height=2.35in,width=3.4in]{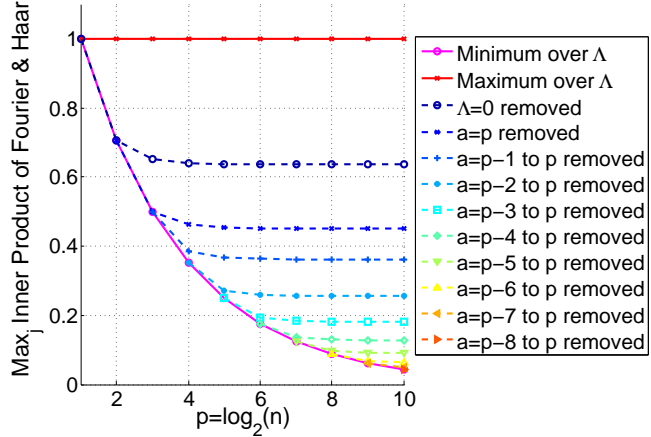}
\caption{The value of $\max_{j\in\{0,\dots,n-1\}}| \inner{\mathbf{f}_j, \H^*_{\Lambda}} |$ is displayed against the (log of the) signal dimension $n=2^p$. The solid magenta curve shows the optimization when minimizing over all possible supports $\Lambda\in\{0,\dots,n-1\}$, given by (\ref{eqn::lowerboundFH}). The solid red curve shows the opposite optimization when maximizing over all possible supports $\Lambda\in\{0,\dots,n-1\}$, given by (\ref{eqn::upperboundFH}). 
The remaining dashed curves show the optimization when maximizing over all supports $\Lambda$ except those in the blocks indicated. }
\label{fig::ubfhinv}
\end{figure}

\section{Discussion}\label{sec::summary}

Adaptive sensing has tremendous potential to improve the accuracy of sparse recovery in a variety of settings.  However, in many practical applications one does not have the freedom to choose arbitrary measurement vectors, but instead must choose from a specified pool of measurements.  One example of particular interest is the setting where measurements must be taken from the Fourier ensemble, as is the case in many medical imaging applications. In this paper we established fundamental limitations on the improvements offered by adaptivity in this setting for certain sparsity bases.  On the other hand, we argued that for other sparsity bases (such as the Haar wavelet basis) the role of adaptivity in the constrained setting is much less straightforward.  We developed a sampling scheme which uses a simple optimization procedure to select measurements adapted to the signal support.  This scheme results in significant improvements once an accurate estimate of the support is obtained, which in practice can be achieved by first dedicating a portion of the measurements to support estimation.   Though this approach is not necessarily provably optimal, it nonetheless demonstrates the potential of adaptive sensing in the constrained setting.  We believe future work in this area can further the understanding of both the limitations of this approach as well as the potential benefits.

\section*{Acknowledgment} 
The authors would like to thank Yaniv Plan and Ran Zhao for stimulating
discussions and help in our theoretical analysis, as well as the reviewers for
their helpful feedback.

\frenchspacing
\bibliographystyle{plain}
\bibliography{bib}

\end{document}